\documentclass[12pt]{article}

\usepackage[margin=2.5cm]{geometry}

\usepackage{amsmath,amsthm,amsfonts,amscd,amssymb,bbm,mathrsfs,enumerate,url}
\numberwithin{equation}{section}
\usepackage{authblk}

\usepackage{graphicx,tikz}
\usepackage[pdfborder={0 0 0}]{hyperref}
\usetikzlibrary{matrix,arrows}
\usetikzlibrary{decorations.pathreplacing}
\usetikzlibrary{decorations.pathmorphing}
\usetikzlibrary{patterns,fadings}
\allowdisplaybreaks

\newcounter{2dCN}
\newcounter{1dCN}
\newcounter{2dW}
\newcounter{1dW}

\def\semicolon{;}
\def\applytolist#1{
    \expandafter\def\csname multi#1\endcsname##1{
        \def\multiack{##1}\ifx\multiack\semicolon
            \def\next{\relax}
        \else
            \csname #1\endcsname{##1}
            \def\next{\csname multi#1\endcsname}
        \fi
        \next}
    \csname multi#1\endcsname}

\def\calc#1{\expandafter\def\csname c#1\endcsname{{\mathcal #1}}}
\applytolist{calc}QWERTYUIOPLKJHGFDSAZXCVBNM;
\def\bbc#1{\expandafter\def\csname bb#1\endcsname{{\mathbb #1}}}
\applytolist{bbc}QWERTYUIOPLKJHGFDSAZXCVBNM;
\def\bfc#1{\expandafter\def\csname bf#1\endcsname{{\mathbf #1}}}
\applytolist{bfc}QWERTYUIOPLKJHGFDSAZXCVBNM;
\def\sfc#1{\expandafter\def\csname s#1\endcsname{{\sf #1}}}
\applytolist{sfc}QWERTYUIOPLKJHGFDSAZXCVBNM;
\def\rfc#1{\expandafter\def\csname r#1\endcsname{{\mathrm #1}}}
\applytolist{rfc}QWERTYUIOPLKJHGFDSAZXCVBNM;
\def\scfc#1{\expandafter\def\csname sc#1\endcsname{{\mathscr #1}}}
\applytolist{scfc}QWERTYUIOPLKJHGFDSAZXCVBNM;

\DeclareMathOperator{\id}{id}
\DeclareMathOperator{\Hom}{Hom}
\DeclareMathOperator{\Ad}{Ad}
\DeclareMathOperator{\dom}{Dom}
\DeclareMathOperator{\supp}{supp}

\DeclareMathOperator{\diff}{Diff_+}
\DeclareMathOperator{\LR}{LR}

\usepackage{xspace}
\DeclareRobustCommand{\eg}{e.g.\@\xspace}

\DeclareRobustCommand{\cf}{cf.\@\xspace}

\DeclareRobustCommand{\ie}{i.e.\@\xspace}

\def\bb1{\mathbbm{1}}

\def\<{\langle}
\def\>{\rangle}
\def\Mob{\mathrm{M\ddot{o}b}}
\def\fin{\text{fin}}

\newtheorem{theorem}{Theorem}[section]

\newtheorem{proposition}[theorem]{Proposition}
\newtheorem{lemma}[theorem]{Lemma}
\theoremstyle{remark}
\newtheorem{remark}[theorem]{Remark}

\title{Wightman fields for two-dimensional conformal field theories with pointed representation category}
\author[1]{Maria Stella Adamo\thanks{{\tt adamoms@ms.u-tokyo.ac.jp}}}
\author[2]{Luca Giorgetti \thanks{{\tt giorgett@mat.uniroma2.it}}}
\author[2]{Yoh Tanimoto\thanks{{\tt hoyt@mat.uniroma2.it}}}

\affil[1]{Department of Mathematical Sciences, The University of Tokyo, \authorcr
3-8-1 Komaba, Tokyo, 153-8914, JAPAN}
\affil[2]{Dipartimento di Matematica, Universit\`a di Roma Tor Vergata,\authorcr
   Via della Ricerca Scientifica 1, I-00133 Roma, Italy}
\date{}

\begin{document}

\maketitle

\begin{center}
 \textit{Dedicated to Roberto Longo in occasion of his 70th birthday}
\end{center}

\begin{abstract}
Two-dimensional full conformal field theories have been studied in various mathematical frameworks,
from algebraic, operator-algebraic to categorical. In this work, we focus our attention on theories with chiral components having pointed braided tensor representation subcategories, namely having automorphisms whose equivalence classes necessarily form an abelian group.
For such theories, we exhibit the explicit Hilbert space structure and construct primary fields as Wightman fields for the two-dimensional full theory.
Given a finite collection of chiral components with automorphism categories with trivial total braiding,
we also construct a local extension of their tensor product as a chiral component.
We clarify the relations with the Longo--Rehren construction, and 
illustrate these results with concrete examples including the $\rU(1)$-current.
\end{abstract}

\section{Introduction}
Two-dimensional conformal field theories (CFTs) have been studied extensively \cite{DMS97} and have attracted the interest of mathematicians
for their algebraic, analytic and geometric structures. In particular, the conformal symmetry in two-dimensional Minkowski spacetime
is described by the diffeomorphism group of the lightrays, therefore it is infinite-dimensional.
This allows to study first the theories that depend only on one of the lightray coordinates (the chiral components)
and then their two-dimensional (full) extensions.

From the operator-algebraic point of view (Haag--Kastler axioms), a general quantum field theory can be formulated as a net of von Neumann algebras
associated with the open spacetime regions \cite{Haag96}. The relations between the full theory and the chiral components
have been obtained, \eg, in \cite{Rehren00}, and they resulted in classification schemes for certain classes
of two-dimensional CFTs, see, \eg, \cite{KL04-2}, \cite{BKL15}. In this course, it was important that from two copies of a single chiral theory
and a family of charged sectors, one can construct two-dimensional theories as extensions of the two-dimensional theory
obtained just by taking tensor products of the chiral components, and then extended it by introducing non-chiral (bulk) fields. Among such extensions,
the most studied ones are the
\lq\lq diagonal" extensions, obtained by letting the chiral theories act simultaneously on a direct sum of copies of their vacuum representation,
and then introducing \lq\lq charged" fields that mix the different components.

Since \cite{LR95}, in the operator-algebraic setting, the (finite index, \ie, \lq\lq relatively small") extensions can be equivalently
well described by Q-systems (\ie, $C^*$-Frobenius algebra objects) in the unitary braided tensor representation category of the net that one wants to extend. Locality of the extension can also be characterized by means of a commutativity condition on the associated Q-system. In the operator-algebras context, the Q-system associated with a finite index \lq\lq diagonal" extension is called a Longo--Rehren Q-system.

It is worth mentioning that the method of Q-systems applies both to chiral and to two-dimensional theories (in fact even in four-dimensions), and that similar ideas (commutative Frobenius algebra objects in tensor categories) emerged independently in other algebraic approaches to CFT, such as vertex operator algebras \cite{HKL15}, \cite{KO02}, both in one and two dimensions, see, \eg, \cite{FuRS02-1}, \cite{HK07FullFieldAlgebras}, \cite{Kong07}, \cite{RFFS07}. In the unitary VOA \cite{CKLW18}, \cite{DL14}, \cite{Gui22} and unitary tensor category context \cite{GLR85}, \cite{DR89}, \cite{LR97}, these notions (Frobenius algebras, $C^*$-Frobenius algebras, Q-systems, to describe extensions) have been recently shown to be equivalent \cite{CGGH23-online}.

The relationship between chiral and full two-dimensional CFTs (not restricted to a single Minkowski spacetime) can also be cast and studied \cite{BGS22} in the more general categorical/operadic framework of locally covariant AQFT \cite{BFV03}, \cite{BSW19}, \cite{BSW21}.

Let us note, however, that most of the purely algebraic frameworks, see, \eg, \cite{HK07FullFieldAlgebras}, are designed particularly for conformal field theories,
and do not apply, as they are, to massive theories.
On the other hand, it is natural to expect that a two-dimensional full CFT can be described in
a more traditional framework for quantum field theory, the Wightman axioms \cite{SW00}.
Wightman fields are operator-valued distributions on a Hilbert space, and a reasonable
description of such fields and the Hilbert space would be desirable.

Having explicit Wightman fields is not just interesting on its own, but could be a starting point
for constructing non-conformal field theory, \eg, by perturbing the dynamics of the CFT by these fields.
Such an idea is presented in \cite{Zamolodchikov89-1}, where certain massive integrable fields
are associated with charged fields in CFT.
Perturbing the dynamics of the free field on the same Hilbert space has been carried out on the de Sitter space \cite{BJM23}.
Therefore, developing a theory of Wightman fields in two dimensions will be a basis for rigorously studying the relations between CFT and massive models \cite{JT23Towards}.

In this paper, we study full two-dimensional CFTs whose chiral components admit a (finite or infinite) collection of automorphisms (invertible objects in the language of tensor categories) among their irreducible representations (superselection sectors). We will define and construct both two-dimensional conformal Haag--Kastler nets and two-dimensional conformal Wightman fields explicitly in terms of the Hilbert space of the chiral components and of their charged sectors. We introduce charged primary fields as operators between different charged sectors of the chiral components, and combine them to obtain local bulk fields.

To be more specific, let $\cA_\rL$ and $\cA_\rR$ be chiral conformal nets on $S^1$ admitting a family of automorphisms (including the defining vacuum representation) among their charged representations (in the sense of Doplicher--Haag--Roberts \cite{DHR69I}, \cite{DHR69II}, \cite{DHR71} \cite{DHR74}, but in one and two dimensions instead of four) parametrized by the same abelian group $G$.
These automorphisms, respectively denoted by $\rL(g)$ and $\rR(g)$, $g\in G$, are defined on the same Hilbert spaces, but we denote them
as $\cH_\rL^{\rL(g)}, \cH_\rR^{\rR(g)}$ to distinguish the representation.
On the Hilbert space $\bigoplus_{g \in G} \cH_\rL^{\rL(g)}\otimes \cH_\rR^{\rR(g)}$,
the chiral observables $\cA_\rL \otimes \cA_\rR$ act diagonally. We will add charged fields that shift the sectors in the direct sum,
under some condition on their braiding (which is satisfied in many cases, see, \eg, the end of Section \ref{U(1)current})
and obtain full two-dimensional conformal nets and conformal Wightman fields.
Moreover, with a similar technique, we construct some (presumably new) local conformal nets
on $S^1$ by taking the tensor product of a family of conformal nets and then extending it, in such a way that a certain trivial total braiding condition is fulfilled.
This generalizes the well-known extensions of a single chiral $\rU(1)$-current net \cite{BMT88},
and it can be seen as a variation of the \lq\lq gluing" construction due to \cite{CKM22Gluing} for VOAs.

We describe explicitly the fields for the $\rU(1)$-current algebra.
Charged primary fields of a single chiral component are given as formal series between charged sectors \cite{TZ12}, \cite{Toledano-LaredoThesis}.
We expect that our construction works for loop group nets (for a simply laced, simple, simply connected, compact group)
at level $1$ as well \cite{Wassermann98}, \cite{Toledano-LaredoThesis}.
With a more involved combination of left and right chiral components, it should also be possibile
to generalize it to other completely rational nets (when the charged fields are available), in the presence of irreducible representations
with statistical dimension greater that 1 (\ie, in the non-pointed tensor category case).

This paper is organized as follows.
In Section 2, we recall the fact that two-dimensional conformal field theories extend to
the Einstein cylinder, then set out the operator-algebraic formulation of the chiral components
and of the full two-dimensional CFTs. We also collect some facts about representations and charged fields
of chiral components, in the case of automorphisms. In Section 3, from a family of chiral components equipped with a collection
of automorphisms and satisfying certain conditions on the braiding, we construct an extension of
their tensor product on $S^1$. In Section 4, we take a pair of left and right chiral components
with a collection of automorphisms and assume that their braidings cancel in a certain sense,
and we construct a full two-dimensional CFT extending the tensor product of the chiral components.
In Section 5, we study the case where the charged primary fields are given as formal series.
Under similar assumptions on the braiding and assuming energy bounds, we exhibit the Wightman fields responsible
for the extensions at the level of nets constructed in the previous sections.
In Section 6, we consider the explicit example of the $\rU(1)$-current.
We exhibit its superselection charge structure, the associated braiding and charged fields, and we show that they
fit in the general construction of the previous sections. In Section 7, we summarize our outlook.

\subsection*{Notations}

In order to keep uniform notations throughout the paper,
we label chiral (one-dimensional) objects, while
we use symbols without a label for two-dimensional objects.
\begin{itemize}
\item $\cA_\kappa, \phi_{\kappa,j}, \psi_\kappa, \cA_\rL, \cA_\rR, \psi_\rL, \psi_\rR$: chiral net/chiral fields/chiral charged fields.
\item $\cA_K, \psi_K$: chiral extension/chiral charged local fields.
\item $\cA$: two-dimensional net.
\item $\widetilde\cA, \widetilde\psi$: two-dimensional net/two-dimensional charged local fields.
\end{itemize}

\section{Preliminaries}

\subsection{Einstein cylinder}\label{cylinder}
As suggested by \cite{LM75}, it is natural to study four-dimensional conformally covariant Wightman fields on the Einstein cylinder\footnote{More
precisely, \cite{LM75} showed that, assuming that there are Wightman fields and their Euclidean $n$-point functions
are invariant under Euclidean conformal group, the representation of the Poincar\'e group can be extended to $\mathrm{SO}(4,2)$.
See \cite[Section 8.3]{KQR21Distributions} for a recent review.}.
Moreover, the conformal group acts on the Einstein cylinder in a natural way and the fields are covariant with respect to its action.
In two dimensions, an analogous procedure can be carried out with one more step because the two-dimensional Einstein cylinder is not simply connected.
To state these results, let us first discuss the conformal geometry.
We follow \cite{KL04-2}, and start the discussion with the lightrays $\bbR$ in $\bbR^{1+1}$.

The M\"obius group $\rP\rS\rL(2,\bbR) \cong \Mob$ acts on the one-point compactification $S^1$ of $\bbR$, where $\bbR$ is identified with $S^1 \setminus \{-1\}$
through the stereographic projection.
Hence, by lifting such action, its universal covering group $\overline{\Mob}$
acts on $\bbR$, the universal covering of $S^1$.
The original lightray $\bbR$ is identified with
the interval $(-\pi, \pi)$ in $\bbR$ as the universal covering of $S^1$. For an interval $I$ such that $\overline I \subset (-\pi, \pi)$,
there is a neighborhood $\cU$ of the unit element in $\Mob$ such that if $\gamma \in \cU$, then $\gamma \cdot I \subset (-\pi,\pi)$.
In this sense, the group $\Mob$ acts locally on the lightray.

The two-dimensional Minkowski space $\bbR^{1+1}$ has the metric $(a,b) = a_0 b_0 - a_1 b_1$, where $a,b \in \bbR^{1+1}$.
With the lightcone coordinates $(a_-, a_+) = \left(\frac{a_0-a_1}{\sqrt2},\frac{a_0+a_1}{\sqrt 2}\right)$,
the Minkowski space is the product of two lightrays $\bbR^{1+1} = \bbR \times \bbR$
(in the lightcone coordinates, not the $(a_0,a_1)$-coordinates) and the metric can be written as $(a,b) = a_- b_+ + a_+ b_-$.
A conformal transformation of $\bbR^{1+1}$ is, by definition, a transformation of $\bbR^{1+1}$
that preserves the metric up to a scalar. From the above expression, it is clear that
a product of any pair of orientation-preserving diffeomorphisms of lightrays is a conformal transformation.

The diamond $D_0 = \{(a_-, a_+): -\pi < a_\pm < \pi\}$ can be mapped to the Minkowski space $\bbR^{1+1}$ 
by the conformal transformation $(a_-, a_+) \mapsto (\tan(\frac12 a_-), \tan(\frac12 a_+))$.
Through this transformation, $\bbR^{1+1}$ can be identified with $D_0$.
On $D_0 = (-\pi, \pi)\times (-\pi,\pi)$, the group $\overline{\Mob} \times \overline{\Mob}$ acts locally in the sense above.

Let $R_t$ be the lift of the rotation by $t$ in $\overline{\Mob}$. 
In a two-dimensional conformal field theory, the correlation functions are invariant under the local action
of $\overline{\Mob} \times \overline{\Mob}$, and moreover, the spacelike
$2\pi$-rotations $\mathfrak R := \{R_{2n\pi}\times R_{-2n\pi}: n \in \bbZ\}$ are often trivial. If this holds, the field theory can be
extended to the Einstein cylinder\footnote{This is topologically equivalent to $S^1 \times \bbR$,
but the product structure is different from the lightray decomposition.}
$\cE = \bbR^{1+1}/\mathfrak R$.
In the most favorable case, the theory extends to $\cE$ and there is also a local action of
the \textbf{conformal group} $\scC = \overline{\diff(S^1)}\times \overline{\diff(S^1)}/\mathfrak R$.

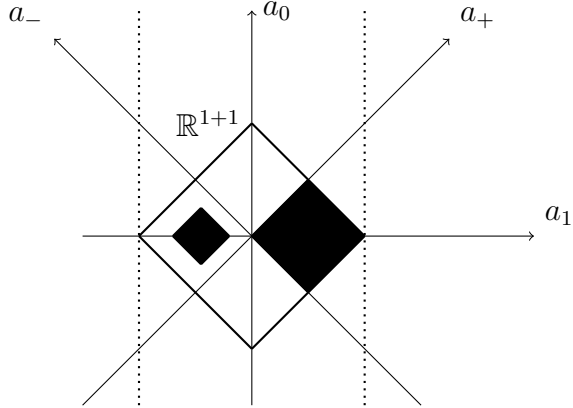
\begin{figure}[ht]\centering
\begin{tikzpicture}[scale=0.75]
        \begin{scope}
        \draw [->] (-3,0) --(5,0) node [above right] {$a_1$};
         \draw [->] (0,-3)--(0,4) node [ right] {$a_0$};
          \draw [thick] (-2,0)-- (0,2) node [ left] {$\bbR^{1+1}$};
          \draw [thick] (0,2)-- (2,0);
          \draw [thick] (-2,0)-- (0,-2);
          \draw [thick] (2,0)-- (0,-2);
          
           \draw [ thick] (0,0)-- (1,1);
          \draw [thick] (0,0)-- (1,-1);
          
\draw [ thick] (-0.4,0) -- (-0.9,0.5);
\draw [thick] (-0.4,0) -- (-0.9,-0.5);
\draw [thick] (-1.4,0) -- (-0.9,0.5);
\draw [thick] (-1.4,0) -- (-0.9,-0.5);
          
          \draw [thick,dotted] (-2,-3)--(-2,4);
          \draw [thick,dotted] (2,-3)--(2,4);
          
          \draw [thin,->] (-3,-3) -- (3.5,3.5) node [above right] {$a_+$};
          \draw [thin,->] (3,-3) -- (-3.5,3.5) node [above left] {$a_-$};
  \fill [color=black,opacity=0.2]
               (-0.4,0) -- (-0.9,0.5) -- (-1.4,0) --(-0.9,-0.5);
           \fill [color=black,opacity=0.6]
              (0,0) -- (1,1) -- (2,0) -- (1,-1);
        \end{scope}
\end{tikzpicture}
\caption{The Minkowski space $\bbR^{1+1}$, depicted as the diamond.
When the spacelike rotation $R_{2\pi} \times R_{-2\pi}$ is trivial, the dotted lines are identified
and it is a subset of the Einstein cylinder $\cE$.}
\label{fig:conformal}
\end{figure}

\subsection{Chiral components and representation theory}\label{chiral}
\subsubsection{Conformal net on \texorpdfstring{$S^1$}{S1}}\label{net1d}
Among two-dimensional conformal fields, there are those that do not depend on one of the lightray coordinates.
They are called \textbf{chiral fields}, and can be restricted to the lightray, then extended to the circle $S^1$ by locality.
We put the index $\kappa$ to the objects in this section, although we study a single chiral field theory at a time.
Note that we do not use the index $\kappa$ to distinguish between left and right chiral components.
This convention will be useful later when we combine different chiral fields on either component of a two-dimensional conformal field theory,
while we denote two-dimensional objects without index.

Let us start with a net on $\bbR$, and see how it extends to $S^1$.
Note that if a unitary projective $U_\kappa$ representation of $\overline{\diff(S^1)}$ on a Hilbert space $\cH_\kappa$ is restricted to
the semisimple subgroup $\overline{\Mob}$, there is a unique true (non-projective) representation
whose quotient in $\rP\rU(\cH_\kappa)$ coincides with $U_\kappa$. In this sense, we can consider the spectrum
of $U_\kappa$ restricted to $\overline{\Mob}$ without ambiguity \cite[Theorem 7.1]{Bargmann54}.

We call a triple $(\cA_\kappa, U_\kappa, \Omega_\kappa)$ a \textbf{conformal net on $\bbR$}
if $\cA_\kappa$ assigns to each open non-dense non-empty interval $I \subset \bbR$ a von Neumann algebra $\cA_\kappa(I)$ on a Hilbert space $\cH_\kappa$,
$U_\kappa$ is a unitary projective representation of $\overline{\diff(S^1)}$ and $\Omega_\kappa \in \cH_\kappa$ such that
\begin{enumerate}[{(1dCN}1{)}]
 \item \textbf{Isotony: }if $I_1\subset I_2$, then $\cA_\kappa(I_1)\subset\cA_\kappa(I_2)$. \label{1dcn:isotony}
 \item \textbf{Locality:} if $I_1$ and $I_2$ are disjoint, then $\cA_\kappa(I_1)\subset\cA_\kappa(I_2)'$. \label{1dcn:locality}
 \item\textbf{Diffeomorphism covariance:} \label{1dcn:diff}
 For a bounded interval $I \subset \bbR$, there is a neighborhood $\cU$ of the unit element of $\diff(S^1)$
 such that if $\gamma \in \cU$ then $\gamma \cdot I \subset \bbR$ and
 \[
  U_\kappa(\gamma)\cA_\kappa(I)U_\kappa(\gamma)^*=\cA_\kappa(\gamma \cdot I).
 \]
 Furthermore, if $\supp \gamma$ is disjoint from $I$, then $\Ad U_\kappa(\gamma)(x) = x$ for $x \in \cA_\kappa(I)$.
 \item\textbf{Positivity of energy:}
 the restriction of $U_\kappa$ to the translation subgroup $\bbR \subset \overline{\Mob}$
 has the spectrum contained in $\bbR_+$.
 \label{1dcn:positiveenergy}
 \item \textbf{Vacuum and the Reeh-Schlieder property:} there exists a unique (up to a phase)
 vector $\Omega_\kappa \in\cH_\kappa$ such that $U_\kappa(g)\Omega_\kappa=\Omega_\kappa$ for $g\in \overline{\Mob}$
 and $\overline{\cA_\kappa(I)\Omega_\kappa}=\cH_\kappa$. \label{1dcn:vacuum}
 \setcounter{1dCN}{\value{enumi}}
\end{enumerate}
If we assume only the covariance with respect to $\overline{\Mob}$, we call it a \textbf{M\"obius-covariant} net.

If $U_\kappa$ factors through $\diff(S^1)$ as a projective representation, that is, if $U_\kappa(R_{2n\pi}), n \in \bbZ$ is a scalar,
then $\cA_\kappa$ can be extended to a net defined on the set $\cI$ of non-dense, non-trivial open intervals on $S^1$
and it satisfies the usual axioms of \textbf{conformal net on $S^1$}, that is a triple $(\cA_\kappa, U_\kappa, \Omega_\kappa)$
satisfying isotony, locality, $\diff(S^1)$-covariance, positivity of energy and the vacuum properties.
See, \eg, \cite[Section 2.1]{KL04-1}, \cite[Chapter 3]{CKLW18}.

We also recall that $U_\kappa$ can be made into a unitary multiplier representation (rather than projective) of $\overline{\diff(S^1)}$,
that is, $U_\kappa(\gamma_1)U_\kappa(\gamma_2) = c(\gamma_1, \gamma_2)U_\kappa(\gamma_1\gamma_2)$ for some $c(\gamma_1, \gamma_2) \in \bbC$
\cite[Theorem A.2]{Carpi04}, \cite[Proposition 5.1]{FH05}. Such $c$ is called the \textbf{cocycle} of the multiplier representation.

\begin{proposition}\label{pr:1dcircle}
 Let $(\cA_\kappa, U_\kappa, \Omega_\kappa)$ be a conformal net on $\bbR$.
 Assume furthermore that $U_\kappa(R_{2\pi}) = \bb1$ (in $\rP\rU(\cH_\kappa)$), that is,
 the lift of $2\pi$-rotation is trivial.
 Then $(\cA_\kappa, U_\kappa, \Omega_\kappa)$ extends to $S^1$.
\end{proposition}
\begin{proof}
 If $U_\kappa(R_{2\pi}) = \bb1$, then its adjoint action is trivial,
 and the natural extension of $\cA_\kappa$ on $(-\pi, \pi)$ to $\bbR$ (the universal covering of $S^1$)
 is periodic, hence we can regard it as a net on $S^1$.
 It is known that the positivity of energy restricted to the rotation subgroup
 and that to the translation subgroup are equivalent (see, \eg, \cite[Lemma 3.1]{Weiner06}).
 The remaining axioms of \cite[Section 2.1]{KL04-1} follow easily.
\end{proof}

The assumption $U_\kappa(R_{2\pi}) = \bb1$ is necessary, but we do not know whether this follows from
more general assumptions (\cf the Bisognano--Wichmann property is necessary and sufficient \cite[Theorem 1.4]{GLW98},
but we are not aware whether it is automatic for conformal nets on $\bbR$).

There is an example of a net on $\bbR$ with a weaker covariance and not extending to $S^1$ ($U_\kappa$ is only a projective representation of
the group generated by translations, dilations and diffeomorphisms of $\bbR$ with compact support):
the $\rU(1)$-current net with the perturbed stress-energy tensor \cite{BS90}, see also the discussion in \cite[Section 5.2]{MT18}.

\subsubsection{Representations of chiral conformal nets}\label{representations}

A \textbf{representation} of a conformal net $\cA_\kappa$
(or more precisely of $(\cA_\kappa, U_\kappa, \Omega_\kappa)$)
on $S^1$ is a family of representations $\rho = \{\rho_I\}$
of $\{\cA_\kappa(I)\}_{I \in \cI}$ on a single Hilbert space $\cH_\kappa^\rho$,
which is compatible in the sense that if $I_1 \subset I_2$, then $\rho_{I_2}|_{\cA_\kappa(I_1)} = \rho_{I_1}$.
If $\cH_\kappa^\rho$ is equal to $\cH_\kappa$ and $\rho_I(\cA_\kappa(I)) = \cA_\kappa(I)$, equivalently if the $C^*$-tensor categorical or statistical dimension $d_\rho$ equals $1$, we say that $\rho$ is an \textbf{automorphism} of $\cA_\kappa$.

Let us summarize the general theory of representations in the case of automorphisms (\cf \cite[Section 6]{MTW18}):
\begin{itemize}
 \item Any automorphism is irreducible,
 hence diffeomorphism covariant in the following sense (\cf \cite[Theorem 6]{DFK04} for projective representations,
 \cite[Section 3.2]{CDVIT21} for local multiplier representations):
 there is a local unitary multiplier representation $U_\kappa^\rho$
 of a neighborhood $\cU$ of the unit element of $\overline{\diff(S^1)}$ with the same cocycle as that of $U_\kappa$
 such that $U_\kappa^\rho(\gamma) = \rho(U_\kappa(\gamma))$ if $\supp \gamma \subset I$ for some $I$
 and $\Ad U_\kappa^\rho(\gamma)(\rho(x)) = \rho(\Ad U_\kappa(\gamma)(x))$ for $x \in \cA_\kappa(I)$ for some $I$, $\gamma \in \cU$
 (the last restriction is why we call it a \textit{local} representation).
 $U_\kappa^\rho$ extends to $\overline{\diff(S^1)}$ as a projective representation
 because $\overline{\diff(S^1)}$ is simply connected.
 Moreover, its restriction to $\overline{\Mob}$
 is a true representation and it has automatically positive energy \cite[Theorem 3.8]{Weiner06}.
 Therefore, the lowest eigenvalue of the generator
 of lift rotations $L_0^\rho$ is uniquely determined, and we denote it by $D_{\kappa, \rho}$.
 This is called the \textbf{conformal dimension}.
 As the $2\pi$-rotation is trivial, the spectrum of $L_0^\rho$ is contained in $D_{\kappa, \rho} + (\bbN \cup \{0\})$.
 By the first part of the proof of \cite[Theorem 5.1]{FH05}, the representation $U_\kappa^\rho$ extends to a
 unitary multiplier representation of $\overline{\diff(S^1)}$.
 \item If $\rho, \rho'$ are two representations of $\cA_\kappa$ and there is a unitary $V \in \rU(\cH_\kappa)$
 such that $V\rho_I(x) = \rho_I'(x)V$, for every $x \in \cA_\kappa(I)$ and for all $I\in \cI$
 we say that they are unitarily equivalent. We also say that $V$ is a unitary intertwiner between $\rho$ and $\rho'$.
 
 \item An automorphism is said to be \textbf{localized in $I$} if $\rho_{I'} = \id$,
 where $I'$ is the interior of $S^1 \setminus I$. In this case, $\rho_{I}$ maps $\cA_\kappa(I)$ to itself,
 \ie, it is an endomorphism (in this case an automorphism) of $\cA_\kappa(I)$, by Haag duality on $S^1$, \cite[Theorem 2.19]{GF93}
 (\ie, $\cA_\kappa(I')= \cA_\kappa(I)'$ for every $I\in\cI$, cf. \cite{BS90}).  
 Given $\rho$ localized in $I_1$, one can always find a unitarily equivalent automorphism localized in another interval $I_2$.
 (More generally, every representation of a conformal net on a separable Hilbert space can be localized in any given interval, \cite[Lemma 4.6]{GF93}).
 If $I_1 \cup I_2$ is not dense in $S^1$, one can take a proper interval $I$
 that contains $I_1 \cup I_2$, and by Haag duality on $S^1$, the unitary operator $V$ implementing the equivalence
 belongs to $\cA_\kappa(I)$. Such a $V$ is called a \textbf{charge transporter}.
 
 \item The operator $z_{\kappa,\rho}(\gamma) := U_\kappa(\gamma)U_\kappa^\rho(\gamma)^*$ for $\gamma\in \overline{\diff(S^1)}$ is a charge transporter between $\rho$ and $\rho^\gamma$, where $\rho^\gamma := \Ad U_\kappa(\gamma) \circ \rho \circ \Ad U_\kappa(\gamma^{-1})$ is localized in $\gamma \cdot I$ if $\rho$ is localized in $I$. We call $z_{\kappa,\rho}(\gamma)$ the \textbf{covariance cocycle}, or just cocycle, of $\rho$.

 \item The \textbf{DHR tensor (or monoidal) product} of representations $\rho_1$ and $\rho_2$ is defined by the composition of the associated endomorphisms (automorphisms in this case) localized in each $I$. Note that the tensor product is not given by considering the tensor product Hilbert space, which would not give a well-defined family of representations of the local algebras. 
 In symbols, $(\rho_1)_I \otimes (\rho_2)_I := (\rho_1)_I \circ (\rho_2)_I$, or just $(\rho_1)_I (\rho_2)_I$, for short. The resulting representation is denoted by $\rho_1\otimes\rho_2$, or just by $\rho_1\rho_2$. See, \eg, \cite[Section IV.2]{GF93}. For strongly additive conformal nets $\cA_{\kappa}$ on $S^1$ (\ie, $\cA_\kappa(I_1)\vee\cA_\kappa(I_2) = \cA_\kappa(I)$ for every $I\in\cI$ and $I_1, I_2\in\cI$ arising as connected components of $I\smallsetminus \{p\}$, $p\in I$), or equivalently for conformal nets on $\bbR$ that satisfy Haag duality on $\bbR$, the tensor product can be defined globally, see, \cite[Appendix B]{KLM01}, \cite[Section 2]{GR18}, which is closer in spirit to the original definition of Doplicher--Haag--Roberts \cite{DHR69I}, \cite{DHR69II}. 

 \item Let $\rho_1, \rho_2$ be two automorphisms localized in $I$. In order to define the \textbf{DHR braiding} \cite{DHR71}, \cite{FRS89} between
 $\rho_1$ and $\rho_2$, let us choose a point on $S^1$, identified with the point at infinity by means of the corresponding stereographic projection onto $\bbR$. 
 Take $\widetilde\rho_1, \widetilde\rho_2$ localized in $\widetilde I_1, \widetilde I_2$, respectively,
 such that $\widetilde I_1 \cap \widetilde I_2 = \emptyset$, $\widetilde I_1, \widetilde I_2$ are away from infinity, and take charge transporters $V_1, V_2$
 between $\rho_1$ and $\widetilde\rho_1$, $\rho_2$ and $\widetilde \rho_2$, respectively.
 The operator $\epsilon^{\pm}_{\rho_1, \rho_2} := \rho_2(V_1^*) V_2^* V_1 \rho_1(V_2)$,
 where $\pm$ depends on whether $I_1$ is on the left or right of $I_2$, is a unitary intertwiner 
 between $\rho_1\rho_2$ and $\rho_2\rho_1$ with all the properties of a unitary braiding in a unitary tensor category.
 It does not depend on the choice of $\widetilde\rho_1, \widetilde\rho_2$ or $V_1, V_2$ under the same configuration of $\widetilde I_1, \widetilde I_2$, or on the chosen point at 
 infinity, see, \eg, \cite[Section IV.4]{GF93}. 
 
 \item Let $\rho$ be an automorphism localized in $I$.
 The conjugate automorphism $\overline\rho$ of $\rho$ is the automorphism localized in $I'$ given by
 $j\circ \rho \circ j$, where $j = \Ad J_I$ and $J_I$ is the modular conjugation of $\cA_\kappa(I)$ with respect to the vacuum.
 It follows that $\rho \circ \overline{\rho} \cong \id$, 
 see \cite[Theorem 8.3]{GL92}, \cite[Theorem 2.11]{GL96}.
\end{itemize}

Let $(\cB_\kappa, U_\kappa, \Omega_\kappa)$ be a conformal net on $S^1$ and let $\{\cA_\kappa\}_{I\in \cI}$ be a family of von Neumann subalgebras
$\cA_\kappa(I) \subset \cB_\kappa(I)$ satisfying covariance with respect to $U_\kappa$.
Then, on the subspace $\cH_{\cA_\kappa}=\overline{\bigcup_{I\in \cI}\cA_\kappa(I)\Omega_\kappa}$,
$(\cA_\kappa|_{\cH_{\cA_\kappa}}, U_\kappa|_{\cH_{\cA_\kappa}}, \Omega_\kappa)$ is a M\"obius-covariant net
with respect to $U_\kappa$ because $U_\kappa(\gamma)x\Omega_\kappa = \Ad U_\kappa(\gamma)(x)\Omega_\kappa \in \cA(\gamma \cdot I)\Omega$
for $x \in \cA_\kappa(I)$, thus $U_\kappa(\gamma)$ preserves $\cH_{\cA_\kappa}$.
The restriction of $U_\kappa|_{\cH_{\cA_\kappa}}$ to $\Mob$ often extends to $\diff(S^1)$
and $\cA_\kappa|_{\cH_{\cA_\kappa}}$ is covariant with respect to it, in which case we say that $\cA_\kappa$ is a (conformal) subnet of $\cB_\kappa$
and write $\cA_\kappa \subset \cB_\kappa$ for simplicity. In this case, we say that $\cB_\kappa$ is an \textbf{extension} of $\cA_\kappa$.

\subsubsection{Charged fields associated with automorphisms}\label{charged}

Let $\{(\cA_\kappa, U_\kappa, \Omega_\kappa)\}$ be a conformal net on $S^1$. In this section, we assume that among its irreducible representations there are non-trivial automorphisms.
In the language of tensor categories, the representations of $\cA_\kappa$ contain a (braided) \textbf{pointed tensor subcategory}.

The construction of charged fields and extensions performed in this section is inspired by \cite{DHR69II} and \cite{DR89-1}, \cite{DR90} in the four-dimensional context, \cf \cite[Chapter 3]{Baumgaertel95}, where the unitary tensor categories in question are in addition symmetrically braided (hence dual to a compact group \cite{DR90}). In \cite{DHR69II} the pointed case is investigated, as in our case. See also \cite{Mueger01} for comparison with the one and two-dimensional context.

Fix an interval $I \in \cI$.
Let $\Delta_\kappa$ 
be a choice of mutually inequivalent automorphisms of $\cA_\kappa$ localized in $I$, one for each unitary equivalence class, including the trivial automorphism $\id$ (the defining vacuum representation) of $\cA_\kappa$.
Under the present assumption, the equivalence classes of automorphisms 
form a \emph{discrete (finite or infinite) abelian group} $G$ under class multiplication $[\rho][\sigma] = [\rho\sigma]$
(where $\rho\sigma$ is the composition of automorphisms) and inversion $[\rho]^{-1} = [\rho^{-1}]$.
Let us assume that $G$ is \emph{finitely generated}, whose elements we denote by $g,h,g^{-1},\ldots$, and we denote by $\iota$ the identity element. Let $\kappa(g) \in \Delta_\kappa$, for every $g\in G$, be the previously made choice of automorphisms of $\cA_\kappa$. We assume that $\kappa(\iota) = \id$. Without loss of generality, by suitably changing the localization of the automorphisms inside $I$, we may assume that $\kappa(g) \kappa(h) = \kappa(h) \kappa(g)$ for every $g,h\in G$. Note that we are \emph{not} assuming $\kappa(g)\kappa(h) = \kappa(gh)$, as $\Delta_\kappa$ need \emph{not} be closed under composition and inverses.
For every $g,h\in G$, let $V^{g,h}$ be a unitary intertwiner in $\cA_\kappa(I)$ between $\kappa(gh)$ and $\kappa(g)\kappa(h)$ realizing the equivalence $[\kappa(gh)] = [\kappa(g)][\kappa(h)]$. Namely, $V^{g,h} \kappa(gh)(x) = \kappa(g)(\kappa(h)(x)) V^{g,h}$ for every $x\in\cA_\kappa(I)$. We may assume that $V^{g,h} = V^{h,g}$ and that $V^{g,\iota} = V^{\iota,h} = 1$.

For each $g \in G$, we define the operator $\psi_\kappa^g$ acting on the Hilbert space $\hat \cH_\kappa := \bigoplus_{g\in G} \cH_\kappa^{\kappa(g)}$, where $\cH_\kappa^{\kappa(g)}:=\cH_\kappa$ the vacuum Hilbert space of $\cA_\kappa$, by setting
\begin{align*}
(\psi_\kappa^{g}\Psi)_h :=  V^{g,h} (\Psi)_{gh},
\end{align*}
for every $\Psi \in \hat \cH_\kappa$. We call the $\psi_\kappa^g$ \textbf{charged field operators}. They satisfy the following:
\begin{itemize}
\item Each $\psi_\kappa^g$ is unitary on $\hat \cH_\kappa$ and it maps $\cH_\kappa^{\kappa(h)}$ to $\cH_\kappa^{\kappa(gh)}$ for every $h\in G$.
\item Consider the representation $\hat \kappa(x) := \bigoplus_{h \in G} \kappa(h)(x)$ of $x\in\cA_\kappa(I)$ on $\hat \cH_\kappa$. 
Then
\begin{align*}
(\psi_\kappa^g \hat\kappa(x)\Psi)_h &= V^{g,h}\left(\hat\kappa(x)\Psi\right)_{gh} \\
&= V^{g,h}\kappa(gh)(x)\left(\Psi\right)_{gh} \\
&= \kappa(g) (\kappa(h)(x)) V^{g,h}\left(\Psi\right)_{gh} \\
&= \kappa(h) (\kappa(g)(x)) V^{g,h}\left(\Psi\right)_{gh} \\
&= \kappa(h) (\kappa(g)(x)) \left(\psi_\kappa^g\Psi\right)_{h} \\
&= \left(\hat\kappa(\kappa(g)(x)) \psi_\kappa^g\Psi\right)_h,
\end{align*}
for every $\Psi \in \hat \cH_\kappa$, where we used that $\kappa(h)$ and $\kappa(g)$ commute in the 4th equality.

Therefore, we have the \lq\lq charged field intertwiner" property (in the sense of Doplicher--Roberts \cite{DR72})
together with its conjugate (by substituting $x$ with $\kappa(g)^{-1}(x^*)$):
\begin{align}\label{eq:chargedintertwiner}
\psi_\kappa^g \hat\kappa(x) &= \hat\kappa(\kappa(g)(x)) \psi_\kappa^g, \\
(\psi_\kappa^g)^* \hat\kappa(x) &= \hat\kappa(\kappa(g)^{-1}(x)) (\psi_\kappa^g)^*, \nonumber
\end{align}
for every $x\in\cA_\kappa(I)$, as operators on $\hat \cH_\kappa$.
\item
The covariance cocycles $z_{\kappa(g)}(\gamma)$ for $\gamma\in \overline{\diff(S^1)}$ are unitary charged transporters between $\kappa(g)$ and $\kappa(g)^\gamma$, and they fulfill the following tensoriality property.

\emph{Tensoriality of cocycles}: $z_{\kappa(g)}(\gamma)\kappa(g)(z_{\kappa(h)}(\gamma)) = z_{\kappa(g)\kappa(h)}(\gamma)$, or equivalently written as $z_{\kappa(g)}(\gamma) \otimes z_{\kappa(h)}(\gamma) = z_{\kappa(g) \otimes \kappa(h)}(\gamma)$. 

See \cite[Appendix A]{Longo97}, \cite[Section 7]{DG18}, \cite[Proposition 6.1]{MTW18}.
\item Let $\hat U_\kappa(\gamma) := \bigoplus_{h\in G} U_\kappa^{\kappa(h)}(\gamma)$ be a representation of $\overline{\diff(S^1)}$. 
As in \cite[Theorem 7.7]{DG18}, \cite[Section 6]{MTW18}, we get
\begin{align*}
(\hat U_\kappa(\gamma) \psi_\kappa^g \hat U_\kappa(\gamma)^*\Psi)_h &= U_\kappa^{\kappa(h)}(\gamma) (\psi_\kappa^g \hat U_\kappa(\gamma)^*\Psi)_h \\
&= U_\kappa^{\kappa(h)}(\gamma) V^{g,h} (\hat U_\kappa(\gamma)^*\Psi)_{gh} \\
&= U_\kappa^{\kappa(h)}(\gamma) V^{g,h} U_\kappa^{\kappa(gh)}(\gamma)^*(\Psi)_{gh} \\
&= U_\kappa^{\kappa(h)}(\gamma) U_\kappa^{\kappa(g)\kappa(h)}(\gamma)^*V^{g,h}(\Psi)_{gh} \\
&= U_\kappa^{\kappa(h)}(\gamma)U_\kappa^{\kappa(\iota)}(\gamma)^*U_\kappa^{\kappa(\iota)}(\gamma)  U_\kappa^{\kappa(g)\kappa(h)}(\gamma)^*V^{g,h}(\Psi)_{gh} \\
&= z_{\kappa(h)}(\gamma)^* z_{\kappa(h) \kappa(g)}(\gamma) V^{g,h} (\Psi)_{gh} \\
&= \kappa(h)(z_{\kappa(g)}(\gamma)) V^{g,h} (\Psi)_{gh} \\
&= (\hat \kappa(z_{\kappa(g)}(\gamma)) \psi_\kappa^g \Psi)_{h},
\end{align*}
for every $\Psi \in \hat \cH_\kappa$. We used that $U_\kappa^{\kappa(g)}(\gamma) = \kappa(g)(U_\kappa(\gamma))$ if $\supp \gamma \subset I$ 
in the 4th equality (see Section \ref{representations}), the definition of $z_{\kappa(g)}(\gamma)$ in the 6th equality,
and the tensoriality of the cocycles in the 7th equality.

Therefore, we have the following \lq\lq covariance property" of the charged fields:
for every $\gamma\in \overline{\diff(S^1)}$, as operators on $\hat \cH_\kappa$,
\begin{align}\label{eq:covariancecharged}
\Ad \hat U_\kappa(\gamma)(\psi_\kappa^g) &= \hat \kappa (z_{\kappa(g)}(\gamma)) \psi_\kappa^g, \\
\Ad \hat U_\kappa(\gamma)((\psi_\kappa^g)^*) &= \hat \kappa (z_{\kappa(g)^{-1}}(\gamma)) (\psi_\kappa^g)^*, \nonumber
\end{align}
where we used that
$(\psi_\kappa^g)^*\hat \kappa (z_{\kappa(g)}(\gamma)^*) = \hat \kappa (\kappa(g)^{-1}(z_{\kappa(g)}(\gamma)^*))(\psi_\kappa^g)^*
= \hat \kappa (z_{\kappa(g)^{-1}}(\gamma)) (\psi_\kappa^g)^*$
by the charged field intertwiner property of $(\psi_\kappa^g)^*$ and the tensoriality of cocycle.
\end{itemize}

\begin{remark}
Note that the charged fields constructed above are highly non-canonical, as they depend, \eg, on the choice of $V^{g,h}$ for every $g,h\in G$ (each $V^{g,h}$ is unique up to unitary equivalence). Note also that we are demanding neither $\psi_\kappa^g \psi_\kappa^h = \psi_\kappa^{gh}$, nor ${\psi_\kappa^{g}}^* = \psi_\kappa^{g^{-1}}$ for $g,h \in G$, \ie, the charged field operators need not form a group, \cf \cite[Section 3]{Rehren90}. Nevertheless, by the intertwiner and covariance properties, for every choice of $V^{g,h}$, they give rise to \emph{relatively local} (but possibly neither irreducible, nor factorial) extensions of $\cA_\kappa$.
\end{remark}

Later we shall need the condition $\psi_\kappa^g \psi_\kappa^h = \psi_\kappa^h \psi_\kappa^g$ (established in Proposition \ref{commutingfieldsZZn} below) on the charged fields just defined, in order to produce local extensions of tensor products by \lq\lq gluing" conformal nets. 
Note that $\psi_\kappa^g \psi_\kappa^h = \psi_\kappa^h \psi_\kappa^g$, alone, does not mean locality of the extension. Note also that the group multiplication condition $\psi_\kappa^g \psi_\kappa^h = \psi_\kappa^{gh}$ would imply the commutativity condition $\psi_\kappa^g \psi_\kappa^h = \psi_\kappa^h \psi_\kappa^g$, as $G$ is abelian.

\begin{proposition}\label{commutingfieldsZZn}
In the notation of this section, suppose that $G=\bbZ_n$ whose elements we label by $\{0,1,\ldots, n-1\}$ modulo $n$, and let $\alpha \in \Delta_\kappa$ be an automorphism localized in $I$ such that $\alpha^k \in \Delta_\kappa$ for all $k=0,1,\ldots, n-1$ and $[\alpha^n] = [\id]$.
Choose a unitary $V\in\cA_\kappa(I)$ intertwining $\id$ with $\alpha^n$, \ie, $Vx = \alpha^n(x)V$ for every $x\in\cA_\kappa(I)$.
Let the $V^{g,h}$, $g,h\in G$, be defined by $V^{g,h} := V$ if $g+h \geq n$, where $g,h \in \{0,1,\cdots,n-1\}$ are thought of as representatives,
and by $V^{g,h} := \bb 1$ if $g+h < n$.
Then $\psi_\kappa^g \psi_\kappa^h = \psi_\kappa^h \psi_\kappa^g$
and $(\psi_\kappa^g)^* \psi_\kappa^h = \psi_\kappa^h (\psi_\kappa^g)^*$
for every $g,h \in G$.
\end{proposition}

\begin{proof}
Only in this proof, we use the additive notation $g+h$ for the group operation.

We first note that with our choice $V^{g,h} \kappa(g+h) = \kappa(g)\kappa(h)V^{g,h}$.
By taking the representative $g,h \in \{0,1,\ldots, n-1\}$,
this can be checked by cases depending on whether $g+h \ge n$ or $g+h < n$
and using that $\kappa(g) = \kappa(1)^g$.

From the definition, it is straightforward that
$\psi^1_\kappa$ (where the upper index $1$ is the group element, not the exponent)
is the product of a shift on $\bigoplus_{g\in G} \cH_\kappa^{\kappa(g)} = \bigoplus_{g\in G} \cH_\kappa$
as a Hilbert space followed by $1$ or $V$, and that $\psi_\kappa^g = (\psi_\kappa^1)^g$,
again with $g \in \{0,1,\ldots, n-1\}$.
From this it is clear that $\psi_\kappa^g, (\psi_\kappa^g)^*$ and $\psi_\kappa^h$ commute.
\end{proof}

Now let $G$ be a finitely generated abelian group.
Then it is isomorphic to a finite product of cyclic groups $\bbZ_n$ or $\bbZ$.
Let us assume that $\Delta_\kappa$ is a choice of mutually inequivalent automorphisms of $\cA_\kappa$ localized in $I$
and their equivalence classes form $G$ as a group.
For each of the finite cyclic groups $\bbZ_n$ in $G$, we can choose a generator in $\Delta_\kappa$
localized in a smaller interval $\underline I \subset I$.
For each of these generators, we choose $V^{h,g} \in \cA(\underline I)$ as in Proposition \ref{commutingfieldsZZn},
constructing the operator $\psi^h_\kappa$ for such $h, g \in \bbZ_n$.
We can make this choice for each copy of $\bbZ_n$ by choosing finitely many mutually disjoint smaller intervals inside $I$, 
obtaining commuting operators $\psi^h_\kappa, \psi^g_\kappa$ where $g,h$ are in different finite cyclic groups in $G$.
For each of the copies of the infinite cyclic group $\bbZ$, we can take a generator $1$ and we may assume that
$\kappa(g) = \kappa(1)^g$ for all $g \in \bbZ$,
thus we can set $V^{g,h} := \bb 1$ and $\psi^g_\kappa \psi^h_\kappa = \psi^{gh}_\kappa$ for every $g,h \in \bbZ$.
Altogether, we obtain charged fields $\psi^h_\kappa$ for every $h \in G$, satisfying $\psi^h_\kappa\psi^g_\kappa = \psi^g_\kappa\psi^h_\kappa$.

\begin{remark}
 Even if $G$ is not finitely generated, if there is a choice of $\Delta_\kappa$
 such that $\kappa(g)\kappa(h) = \kappa(gh)$, then one can construct
 $\psi^h_\kappa$ satisfying $\psi^g_\kappa \psi^h_\kappa = \psi^{gh}_\kappa = \psi^h_\kappa \psi^g_\kappa$ as in \cite[Section 6]{MTW18},
 thus $(\psi^g)^* = \psi^{g^{-1}}$  and the rest of our construction works as well.
 This happens with the $\rU(1)$-current net with $G= \bbQ$ or $\bbR$.
\end{remark}

\subsection{Two-dimensional conformal nets}\label{net2d}

In \cite{KL04-2}, it was shown that any Haag--Kastler net on $\bbR^{1+1}$ that is locally conformally covariant extends
to the Einstein cylinder, using a slight modification of the conformal spin-statistics theorem \cite{GL96}, \cite[Theorem A.5]{MT18}.
To state this result precisely, let us start with a net on $\bbR^{1+1}$.
We call a triple $(\cA, U, \Omega)$ a \textbf{conformal net on $\bbR^{1+1}$}
if $\cA$ assigns to each open region $O \subset \bbR^{1+1}$ a von Neumann algebra $\cA(O)$ on a Hilbert space $\cH$,
$U$ is a unitary projective representation of $\overline{\diff(S^1)}\times \overline{\diff(S^1)}$ and $\Omega \in \cH$ satisfying
\begin{enumerate}[{(2dCN}1{)}]
 \item \textbf{Isotony: }if $O_1\subset O_2$, then $\cA(O_1)\subset\cA(O_2)$. \label{2dcn:isotony}
 \item \textbf{Locality:} if $O_1$ and $O_2$ are spacelike separated, then $\cA(O_1)\subset\cA(O_2)'$. \label{2dcn:locality}
 \item\textbf{Diffeomorphism covariance:}
 For a bounded region $O \subset \bbR^{1+1}$, there is a neighborhood $\cU$ of the unit element of $\overline{\diff(S^1)}\times \overline{\diff(S^1)}$
 such that if $\gamma \in \cU$ then $\gamma\cdot O \subset \bbR^{1+1}$ and
 \[
  U(\gamma)\cA(O)U(\gamma)^*=\cA(\gamma\cdot O).
 \]
 Furthermore, if $\supp \gamma$ is disjoint from $O$, then $\Ad U(\gamma)(x) = x$ for $\cA(O)$.
 \label{2dcn:diff}
 \item\textbf{Positivity of energy:}
 the restriction of $U$ to the translation subgroup $\bbR^2 \subset \overline{\Mob} \times \overline{\Mob}$
 has the joint spectrum contained in the closed forward light cone
 $\overline {V_+}=\{(a_0,a_1)\in\bbR^{1+1}: a_0^2-a_1^2\geq0, a_0\geq0\}$. \label{2dcn:positiveenergy}
 \item \textbf{Vacuum and the Reeh-Schlieder property: }there exists a unique (up to a phase)
 vector $\Omega \in\cH$ such that $U(g)\Omega=\Omega$ for $g\in \overline{\Mob}$ 
 and is cyclic for any local algebra, namely $\overline{\cA(O)\Omega}=\cH$. \label{2dcn:vacuum}

 \setcounter{2dCN}{\value{enumi}}
\end{enumerate}
If there is a conformal net on $\bbR^{1+1}$ as above,
we can consider it as a net on the diamond $D_0 = (-\pi, \pi) \times (-\pi, \pi)$ as in Section \ref{cylinder}.
With this identification, the group $\overline{\diff(S^1)}\times \overline{\diff(S^1)}$
acts on $\bbR \times \bbR$, and we can extend the net $\cA$ by covariance.
However, this extension is not very natural.
Indeed, it often happens that the conformal net extends to $\cE$ rather than to $\bbR \times \bbR$.
We say that $(\cA, U, \Omega)$ is a \textbf{conformal net on $\cE$}
if $\cA$ is defined for regions in $\cE$, covariant with respect to $U$
along the natural action of $\overline{\diff(S^1)}\times \overline{\diff(S^1)}$
and local in the sense that $\cA(O_1)$ and $\cA(O_2)$ commute whenever
there is a translation $\gamma$ of the cylinder such that $D_0$ contains $\gamma \cdot O_1, \gamma \cdot O_2$
and they are spacelike there.
A sufficient condition is given in \cite[Theorem A.5]{MT18} (the Bisognano--Wichmann property for wedges),
however, it is not immediate to check it in the examples we construct.
Instead, the following is easier to check and gives immediately the desired extension.

\begin{proposition}\label{pr:2dcylinder}
 Let $(\cA, U, \Omega)$ be a conformal net on $\bbR^{1+1}$.
 Assume furthermore that $U(R_{2\pi}\times R_{-2\pi}) = \bb1$ (in $\rP\rU(\cH)$), that is,
 the spacelike $2\pi$-rotation is trivial.
 Then $(\cA, U, \Omega)$ extends to a conformal net on $\cE$.
\end{proposition}
\begin{proof}
 This is parallel to Proposition \ref{pr:1dcircle}.
 
 As $U(R_{2\pi}\times R_{-2\pi}) = \bb1$, the representation $U$
 factors through the group $\scC$ (see Section \ref{cylinder}).
 Furthermore, as $U(R_{2\pi}\times R_{-2\pi}) = \bb1$,
 for any region $O$, $\cA(O)$ and $\cA(R_{2\pi}\times R_{-2\pi}\cdot O)$ coincide.
 Therefore, we can identify any point $x$ on $\bbR\times \bbR$ and $R_{2\pi}\times R_{-2\pi}\cdot x$,
 and obtain the Einstein cylinder $\cE$.
 Covariance follows by definition, and locality in the sense above follows from
 covariance and locality in $D_0$.
\end{proof}
In the situation of Proposition \ref{pr:2dcylinder},
our net $(\cA, U, \Omega)$ is equivalent to a local conformal net in the sense of \cite[Section 2]{KL04-2}.

Rehren introduced the maximal chiral nets of a two-dimensional conformal net $(\cA, U, \Omega)$ that extends to
$\cE$ \cite{Rehren00}: let $I_\rL \times I_\rR$ be a diamond in $\bbR^{1+1}$.
Define $\cA^{\text{max}}_\rL(I_\rL) = \cA(I_\rL \times I_\rR) \cap U(\iota \times \overline{\diff(S^1)})'$,
where $\iota$ is the unit element of $\overline{\diff(S^1)}$.
Then $\cA^{\text{max}}_\rL$ is \textit{a priori} a conformal net on $\bbR$ defined on the Hilbert subspace
$\cH_\rL = \overline{\cA^{\text{max}}_\rL\Omega}$, and
satisfies the condition of Proposition \ref{pr:1dcircle}, therefore, it extends to $S^1$.
Similarly, one can define $\cA^{\text{max}}_\rR$ on $\cH_\rR$.
The original full net $\cA$ contains the tensor product $\cA^{\text{max}}_\rL\otimes \cA^{\text{max}}_\rR$
on the subspace $\cH_\rL \otimes \cH_\rR \subset \cH$ (more precisely, there is a natural injective homomorphism
from $\cA^{\text{max}}_\rL\otimes \cA^{\text{max}}_\rR$ into $\cA$).
In this sense, a generic two-dimensional conformal net $\cA$ is an extension of
the tensor product net $\cA^{\text{max}}_\rL\otimes \cA^{\text{max}}_\rR$.

\section{One-dimensional gluing from trivial total braiding}\label{net-trivialtotalbraiding}

In this Section, from a family of conformal nets on $S^1$ satisfying certain conditions, we construct \emph{local} extensions of their tensor products.
Let us start with a \emph{finite} collection conformal nets $\{(\cA_\kappa, U_\kappa, \Omega_\kappa)\}_{\kappa \in K}$ on $S^1$ labelled by $\kappa \in K$, $|K|<\infty$.
We fix an interval $I \in \cI$.
We assume that each $\cA_\kappa$ admits a (not necessarily finite) collection $\Delta_\kappa$ of mutually inequivalent and commuting automorphisms localized in $I$, containing the trivial automorphism $\id_\kappa$, and whose fusion rules (up to unitary equivalence) are isomorphic to the \emph{same} abelian group $G$, as in Section \ref{charged}.
This group $G$ should be either a finitely generated abelian group, or we assume that we can choose elements in $\Delta_\kappa$ and charged fields $\psi_\kappa^h$ as in Section \ref{charged}, fulfilling the conclusions of Proposition \ref{commutingfieldsZZn} (for each $\kappa \in K$).
Denote also by $\kappa : G \to \Delta_\kappa$ a bijection (not necessarily an isomorphism, as each $\Delta_\kappa$ need not be closed under composition and inverses), such that $\kappa(\iota) = \id_\kappa$, for each $\kappa\in K$. Here $\kappa$ is used as an index as well as a map, with a slight abuse of notations, hence $\kappa(g)$ is an automorphism of $\cA_\kappa$.
Furthermore, we assume that, for all $g_1, g_2, g \in G$,
\begin{itemize}
\item $\prod_{\kappa} \epsilon^\pm_{\kappa(g_1), \kappa(g_2)} = 1$,
\item $\prod_{\kappa} \epsilon^\pm_{\kappa(g_1)^{-1}, \kappa(g_2)} = 1$,
\item $\sum_\kappa D_{\kappa, \kappa(g)} \in \bbZ$,
\end{itemize}
where the choice of $\pm$ above is common for all $\kappa$, and we denoted by the same symbol $\epsilon^{\pm}_{\kappa(g_1), \kappa(g_2)}$ the phase multiple of $\bb 1$ (the trivial intertwiner between $\kappa(g_1) \kappa(g_2)$ and itself) associated with the braiding of $\kappa(g_1)$ and $\kappa(g_2)$. This makes sense by the commutativity assumption on $\Delta_\kappa$, namely $\kappa(g_1) \kappa(g_2) = \kappa(g_2) \kappa(g_1)$.

Now we construct a conformal net $\cA_K$ on $S^1$ as follows.
\begin{itemize}
 \item The Hilbert space of our net is $\cH_K := \bigoplus_{g \in G} \bigotimes_{\kappa \in K} \cH_\kappa^{\kappa(g)}$.
 On this space, we let any operator of the form $\bigotimes_\kappa x_\kappa \in \bigotimes_\kappa \cA_\kappa(I)$
 act\footnote{To clarify the notation once more, $\kappa$ in $x_\kappa$ and $\Delta_\kappa$ works as an index to distinguish
 different tensor components, while $\kappa$ in $\kappa(g)$ indicates the choice of automorphism in $\Delta_\kappa$ labelled by $g\in G$. 
 Therefore, $\kappa(g)(x_\kappa)$ is a bounded operator on $\cH_\kappa^g$.
 In general, $\kappa$ works as an index except for $\kappa(g)$. \label{fn:kappaindex}}
 as $\bigoplus_{g \in G} \bigotimes_\kappa \kappa(g)(x_\kappa)$.
 Denote this representation of the tensor product net $\bigotimes_\kappa \cA_\kappa$ by $K$.
 \item We also consider an auxiliary tensor product space
 $\hat \cH_K := \bigotimes_{\kappa \in K} \bigoplus_{g \in G} \cH_\kappa^{\kappa(g)} \supset \cH_K$.
 The representation of the element $\bigotimes_\kappa x_\kappa$ on this space
 is denoted by $\hat K(\bigotimes_\kappa x_\kappa)$.
 We have $\hat K(\bigotimes_\kappa x_\kappa)|_{\cH_K} = K(\bigotimes_\kappa x_\kappa)$.
 \item The vacuum vector of the net $\cA_K$ will be $\Omega_K := \bigotimes_\kappa \Omega_\kappa \in \bigotimes_\kappa \cH_{\kappa}^{\kappa(\iota)}$, belonging to the $\iota$-th component of $\cH_K$.
 \item The covariance is given by $U_K(\gamma) := \bigoplus_{g \in G} \bigotimes_\kappa U_\kappa^{\kappa(g)}(\gamma)$,
 which is a unitary multiplier representation of $\overline{\diff(S^1)}$.
 By the condition that $\sum_\kappa D_{\kappa, \kappa(g)} \in \bbZ$, $U_K$ satisfies the assumptions of Proposition \ref{pr:1dcircle}
 and is a
 unitary multiplier representation of $\overline{\diff(S^1)}$.
 Thus by ignoring the phase, it is a unitary projective representation of $\overline{\diff(S^1)}$.
 This also extends naturally to $\hat U_K$ on $\hat \cH_K$.
 \item For each element $g \in G$, we introduce the charged field operator $\psi_K^{g}$ as follows.
 Our Hilbert space $\cH_K = \bigoplus_{g \in G} \bigotimes_\kappa \cH_\kappa^{\kappa(g)}$ 
 is a ``diagonal'' subspace of the auxiliary Hilbert space $\hat \cH_K = \bigotimes_\kappa \hat\cH_\kappa = \bigotimes_\kappa \bigoplus_{g\in G} \cH_\kappa^{\kappa(g)}$
 in a natural way, where recall that $\hat\cH_\kappa$ was defined and used in Section \ref{charged}. 
 Let
 \begin{align*}
  \psi_K^g := \bigotimes_\kappa \psi_\kappa^g,
 \end{align*}
 where $\psi_\kappa^g$ are charged fields acting as in Section \ref{charged} on $\hat \cH_\kappa$,
 hence $\psi_K^g$ acts on $\bigotimes_\kappa \hat\cH_\kappa$
 but preserves $\cH_K$.
Now, for the $\psi_\kappa^g$ we make the choice of the $V^{g,h}$ leading to the conclusions of Proposition \ref{commutingfieldsZZn} and comments thereafter. Namely, we choose charged fields such that $\psi_\kappa^g \psi_\kappa^h = \psi_\kappa^h \psi_\kappa^g$ and $(\psi_\kappa^g)^* \psi_\kappa^h = \psi_\kappa^h (\psi_\kappa^g)^*$ for every $g,h \in G$.

 \item From the charged intertwiner property \eqref{eq:chargedintertwiner}, it follows that 
 \begin{align}\label{eq:chargedfieldchiral}
  \psi_K^g \hat K(\bigotimes_\kappa x_\kappa)
  &= \hat K\left(\bigotimes_\kappa \kappa(g)(x_\kappa)\right) \psi_K^g, \\
  (\psi_K^g)^* \hat K(\bigotimes_\kappa x_\kappa)
  &= \hat K\left(\bigotimes_\kappa \kappa(g)^{-1}(x_\kappa)\right) (\psi_K^g)^*.
 \end{align}

 \item The local algebras are given as follows.
 For the interval $I$ fixed above, we set 
 $$\cA_K(I) := \hat K\left(\bigotimes_\kappa \cA_\kappa(I)\right) \vee \{\psi_K^g, (\psi_K^g)^*\}_{g\in G}.$$
 For any other interval $\gamma \cdot I \subset \bbR$ given by the action of a diffeomorphism $\gamma$,
 we set $\cA_K(\gamma \cdot I) := \Ad U_K(\gamma)(\cA_K(I))$.
 We will show below that this net is well-defined, diffeomorphism covariant, local, and it extends to $S^1$.
\end{itemize}

\begin{theorem}\label{th:1dext}
 Let $\{(\cA_\kappa, U_\kappa, \Omega_\kappa)\}_{\kappa \in K}$ as above and
 assume that $\{\Delta_\kappa\}_{\kappa \in K}$, their braiding and $D_{\kappa, \kappa(g)}$ satisfy
 the conditions stated at the beginning of this section. 
 Then the net $(\cA_K, U_K, \Omega_K)$ satisfies (1dCN\ref{1dcn:isotony})--(1dCN\ref{1dcn:vacuum})
 and the condition of Proposition \ref{pr:1dcircle}, hence provides a conformal net on $S^1$, extending the chiral tensor product net $\bigotimes_\kappa \cA_\kappa$.
\end{theorem}
\begin{proof}
 Although the calculations will be similar to those of \cite[Section 6]{MTW18},
 the setting is different because we construct a chiral net on $S^1$, instead of two-dimensional on $\cE$. Let us present the proofs to exhibit why this works for tensor products of $|K|$ factors.
 
 Let us first compute on the auxiliary Hilbert space $\bigotimes_\kappa \hat\cH_\kappa$ with $\hat U_K(\gamma) := \bigotimes_\kappa \hat U_\kappa(\gamma)$, where $\hat U_\kappa(\gamma) = \bigoplus_{h\in G} U_\kappa^{\kappa(h)}(\gamma)$ was defined and used in Section \ref{charged}. Using the covariance property of charged fields \eqref{eq:covariancecharged} (\ie, the formula for the adjoint action of $\hat U_\kappa(\gamma)$ on $\psi_\kappa^g$ acting on each $\hat\cH_\kappa$), we get
 \begin{align*}
  \Ad \hat U_K(\gamma)(\psi_K^g) &= \hat K \left(\bigotimes_\kappa z_{\kappa(g)}(\gamma)\right)\psi_K^g, \\
  \Ad \hat U_K(\gamma)((\psi_K^g)^*) &= \hat K \left(\bigotimes_\kappa z_{\kappa(g)^{-1}}(\gamma)\right)(\psi_K^g)^*.
 \end{align*} 
 Furthermore, we also have $\psi_K^g\psi_K^h = \psi_K^h \psi_K^g$, $(\psi_K^g)^*\psi_K^h = \psi_K^h (\psi_K^g)^*$
 by the commutation condition we imposed on the charged fields $\psi_\kappa^g$, $\psi_\kappa^h$, and $(\psi_\kappa^g)^*$.
 
 We have to make sure that $\cA_K(\gamma \cdot I)$ is well-defined, that is, the definition does not depend on the choice of $\gamma$.
 This question reduces to whether $\Ad U_K(\gamma)(\cA_K(I)) = \cA_K(I)$ if $\gamma \cdot I = I$.
 This is true because we can check the inclusion for $\gamma$ with compact support and generating elements of $\cA_K(I)$:
 \begin{align*}
  \Ad U_K(\gamma)\left( K\left(\bigotimes_\kappa x_\kappa\right)\right) &=  K\left(\bigotimes_\kappa \Ad U_\kappa(\gamma)(x)\right) \in \cA_K(I), \\
  \Ad \hat U_K(\gamma)(\psi_K^g) &= \hat K\left(\bigotimes_\kappa z_{\kappa(g)}(\gamma)\right)\psi_K^g \in \cA_K(I), \\
  \Ad \hat U_K(\gamma)((\psi_K^g)^*) &= \hat K \left(\bigotimes_\kappa z_{\kappa(g)^{-1}}(\gamma)\right)(\psi_K^g)^* \in \cA_K(I).
 \end{align*}
 because $z_{\kappa(g)}(\gamma) \in \cA_\kappa(I)$. As a general $\gamma$ can be written as a product of diffeomorphisms with compact supports,
 this gives the inclusion $\Ad U_K(\gamma)(\cA_K(I)) \subset \cA_K(I)$, and the converse inclusion is obtained by applying this to $\gamma^{-1}$.

 With this well-definedness, the first part of covariance (1dCN\ref{1dcn:diff}) follows by definition.
 Concerning the second part, if $\supp \gamma$ is disjoint from $I$, then $U_K(\gamma)$ commutes
 both with $\hat K\left(\bigotimes_\kappa \cA_\kappa(I)\right)$ by covariance of $\cA_\kappa$.
 Furthermore, $U_K(\gamma)$ commutes with $\psi_K^g$ because componentwise $\psi_K^g$ is a product of
 the shift and an element in $\bigotimes_\kappa \cA_\kappa(I)$,
 while $U_K(\gamma) \in \hat K\left(\bigotimes_\kappa \cA_\kappa(I')\right)$ and this commutes with the shift
 because the $\kappa(g)$ are localized in $I$.
 
 Positivity of energy (1dCN\ref{1dcn:positiveenergy}) follows because $U_K$ is a direct sum of
 positive-energy representations. The vacuum $\Omega_K$ is invariant under $U_K|_{\Mob}$ because it is the
 tensor product of the vacuum vectors of $\cA_\kappa$. It is cyclic for $\cA_K(I_1)$, where $I_1 \subset I'$,
 because $\cA_K(I_1)\Omega_K$ spans $\bigotimes_\kappa \cH_\kappa$ due to the cyclicity of the vacua of $\cA_\kappa$,
 and the whole $\cH_K$, since $\cA_K(I_1)$ contains the shifts $\hat K(z_g(\gamma))\psi_K^g$
 (up to a unitary on the left). From the assumption $\sum_\kappa D_{\kappa, \kappa(g)} \in \bbZ$ for each $g\in G$, we can apply 
 Proposition \ref{pr:1dcircle}. 
 
 As for locality, we take $\gamma$ such that $\gamma \cdot I$ is disjoint from $I$.
 We have to show that generating elements $\hat K(\bigotimes_\kappa x_\kappa)$, $\psi^{g_1}_K$ of $\cA_K(I)$ and
 $\Ad U_K(\gamma)\left(\hat K(\bigotimes_\kappa y_\kappa)\right)$, $\Ad U_K(\gamma)\left(\psi^{g_2}_K\right)$ of $\cA_K(\gamma \cdot I)$ commute.
 This is easy except the one involving $\psi^{g_1}_K$ and $\Ad U_K(\gamma)(\psi^{g_2}_K)$.
 As we have seen, $\Ad U_K(\gamma)(\psi_K^g) = \hat K\left(\bigotimes_\kappa z_{\kappa(g)}(\gamma)\right) \psi_K^g$.
 Therefore, to compute the commutator,
 \begin{align*}
  \psi_K^{g_1} \Ad U_K(\gamma)(\psi_K^{g_2}) &= \psi_K^{g_1} \hat K\left(\bigotimes_\kappa z_{\kappa(g_2)}(\gamma)\right)\psi_K^{g_2} \\
  &= \hat K\left(\bigotimes_\kappa \kappa(g_1)(z_{\kappa(g_2)}(\gamma))\right)\psi_K^{g_1}\psi_K^{g_2},
  \end{align*}
and
  \begin{align*}
  \Ad U_K(\gamma)(\psi_K^{g_2}) \psi_K^{g_1} &= \hat K\left(\bigotimes_\kappa z_{\kappa(g_2)}(\gamma)\right)\psi_K^{g_2} \psi_K^{g_1}.
  \end{align*}
Now observe that $z_{\kappa(g_2)}(\gamma)^* \kappa(g_1)(z_{\kappa}(g_2)(\gamma)) = \epsilon^{\pm}_{\kappa(g_1),\kappa(g_2)}$ (the DHR braiding, a scalar in the case of automorphisms $\kappa(g_1)$ and $\kappa(g_2)$) where the $\pm$ sign depends only on whether $\gamma$ moves $I$ to the left or to the right,
and the choice of $\pm$ is common for all $\kappa$).
Hence the two left hand sides above are equal if and only if $\prod_\kappa \epsilon^{\pm}_{\kappa(g_1),\kappa(g_2)} = 1$ and $\psi_K^{g_1}\psi_K^{g_2} = \psi_K^{g_2}\psi_K^{g_1}$ for every $g_1, g_2 \in G$.
The commutation between $(\psi_K^{g_1})^*$ and $\Ad U_K(\gamma)(\psi_K^{g_2})$
follows too by replacing $\kappa(g_1)$ by $\kappa(g_1)^{-1}$ and using the condition $\prod_\kappa \epsilon^{\pm}_{\kappa(g_1)^{-1},\kappa(g_2)} = 1$ and $(\psi_K^{g_1})^*\psi_K^{g_2} = \psi_K^{g_2}(\psi_K^{g_1})^*$.
Locality for general intervals follows from the previous paragraph and covariance.
\end{proof}

If all $\cA_\kappa$ are \textit{strongly additive}, that is,
$\cA_\kappa(I_3) = \cA_\kappa(I_1)\vee \cA_\kappa(I_2)$
where $I_1, I_2$ is obtained by removing one point from $I_3$,
then $\cA_K$ is strongly additive, too. This can be seen by taking $I_1 = I$,
then $\cA_K(I_3)$ is generated by the same $\psi_K^g \in \cA_K(I)$ and
$K(\bigotimes_\kappa x_\kappa)$, $x_\kappa \in \cA_\kappa(I_3) = \cA_\kappa(I) \vee \cA_\kappa(I_2)$.

This construction can be carried out even if there is only one index $\kappa$,
if the condition $\epsilon^{\pm}_{\kappa(g_1),\kappa(g_2)} = 1$ is satisfied.
When applied to the $\rU(1)$-current net, this gives the extensions considered in \cite{BMT88}, \cf Section \ref{U(1)current}, where $G=\bbZ$, the map $\kappa$ is a group isomorphism, and $\kappa(g_1)$, $\kappa(g_2)$ are powers of the same fixed automorphism $\kappa(h)$ having $\epsilon^{\pm}_{\kappa(h),\kappa(h)} = 1$ (bosonic automorphism).

\section{Two-dimensional conformal field theory arising from braiding-cancelling map}\label{2dnet}
With an idea similar to that of Section \ref{net-trivialtotalbraiding}, we construct conformal nets on $\bbR^{1+1}$.
Here we start with left and right chiral components, take their tensor product as a two-dimensional conformal net,
then find extensions of it. The problem of constructing extensions on $\bbR^{1+1}$ by adding \lq\lq charged fields" 
associated with suitable subgroups of automorphisms on top of a tensor product theory has also been tackled in \cite[Section 5]{Rehren90-1}, 
see in particular \cite[Proposition 5.5]{Rehren90-1}, within the framework of the reduced field bundle \cite{FRS89}, \cite{FRS92}.

\subsection{Extensions with pointed tensor category}

Let us start with a pair of conformal nets $\cA_\rL, \cA_\rR$ on $S^1$.
We fix an interval $I \subset S^1$.
As in Section \ref{net-trivialtotalbraiding}, we assume that $\cA_\rL$ and $\cA_\rR$ admit, respectively, a family $\Delta_\rL, \Delta_\rR$ (finite or infinite) of mutually inequivalent and commuting automorphisms localized in $I$, containing the trivial automorphism $\id_\rL, \id_\rR$, and whose fusion rules (up to unitary equivalence) are isomorphic to the same abelian group $G$, \cf Section \ref{charged}. This group $G$ should be either finitely generated, or we assume that we can choose elements in $\Delta_\rL, \Delta_\rR$ and charged fields $\psi_\rL^g, \psi_\rR^g$ with the commuting property as in the conclusions of Proposition \ref{commutingfieldsZZn} (for $\kappa = \rL, \rR$) and comments thereafter. In this section, we denote by $\rL, \rR$ (instead of $\kappa$) the bijections from $G$ to $\Delta_\rL, \Delta_\rR$, respectively, corresponding to the choice of automorphisms, and by $D_{\rL(g)}$, $D_{\rR(g)}$, where $g\in G$, the conformal dimensions of the associated automorphisms. 

 We assume that

\begin{itemize}
\item $\Delta_\rL$ and $\Delta_\rR$ contain the trivial automorphisms $\id_\rL, \id_\rR$ of $\cA_\rL$ and $\cA_\rR$, corresponding to the identity element $\iota$ of $G$ via $\rL(\iota) = \id_\rL$, $\rR(\iota) = \id_\rR$.
\item For every $g_1, g_2 \in G$, it holds that $\overline{\epsilon^{\pm}_{\rL(g_1), \rL(g_2)}} = \epsilon^{\mp}_{\rR(g_1), \rR(g_2)}$ and $
\overline{\epsilon^{\pm}_{\rL(g_1), \rL(g_2)^{-1}}} = \epsilon^{\mp}_{\rR(g_1), \rR(g_2)^{-1}}$ (braiding cancellation),
and $D_{\rL(g)} - D_{\rR(g)}\in\mathbb{Z}$ for every $g \in G$. 

\end{itemize}

We denote by the same symbol $\epsilon^{\pm}_{\rL(g_1), \rL(g_2)}$ the phase multiple of $\bb 1$ (the trivial intertwiner between $\rL(g_1) \rL(g_2) = \rL(g_2) \rL(g_1)$ and itself) associated with the braiding of $\rL(g_1)$ and $\rL(g_2)$. Similarly for $\rR(g_1)$ and $\rR(g_2)$.

\begin{remark}
Let $\cA_\rL = \cA_\rR$. If $G$ is cyclic, then the braiding cancellation is automatic if we take $\Delta_\rL = \Delta_\rR$
and $\rL(g) = \rR(g)^{-1}$ for every $g \in G$. Indeed, in general, for commuting automorphisms, it holds that 
$\overline{\epsilon^{\pm}_{\rL(g_1), \rL(g_2)}} = \epsilon^{\mp}_{\rL(g_2)^{-1}, \rL(g_1)^{-1}}$.
Now, for automorphisms $\rL(g_1)$, $\rL(g_2)$ that are powers of the same $\rL(h)$, namely $\rL(g_1) = \rL(h)^{n}$, $\rL(g_2) = \rL(h)^{m}$ for $n,m\in\bbZ$, it holds $\epsilon^{\pm}_{\rL(g_1), \rL(g_2)} = (\epsilon^{\pm}_{\rL(h), \rL(h)})^{nm} = \epsilon^{\pm}_{\rL(g_2), \rL(g_1)}$ (namely, the $nm$-th power of the statistical phase of $\rL(h)$). Hence $g_1$ and $g_2$ can be exchanged and we have braiding cancellation.
\end{remark}

With this input, we construct a conformal net $\widetilde\cA$ on $\bbR^{1+1}$ as follows.
\begin{itemize}
\item The full Hilbert space is $\widetilde\cH := \bigoplus_{g \in G} \cH_\rL^{\rL(g)}\otimes \cH_\rR^{\rR(g)}$.
On $\widetilde\cH$, any operator of the form $x_\rL\otimes x_\rR \in \cA_\rL(I_\rL)\otimes \cA_\rR(I_\rR)$
acts as $\bigoplus_{g \in G} \rL(g)(x_\rL)\otimes \rR(g)(x_\rR)$
(see the footnote \ref{fn:kappaindex} for this notation).
Denote this representation by $\widetilde\tau$.
\item The vacuum vector of $\widetilde \cA$ is $\Omega_\rL\otimes \Omega_\rR \in \cH_\rL^{\rL(\iota)}\otimes \cH_\rR^{\rR(\iota)}$.
\item The covariance is given by
$\widetilde U(\gamma_\rL\times \gamma_\rR) := \bigoplus_{g \in G} U_\rL^{\rL(g)}(\gamma_\rL)\otimes U_\rR^{\rR(g)}(\gamma_\rR)$,
which is a unitary projective representation of $\overline{\diff(S^1)}\times \overline{\diff(S^1)}$.
By the assumption that $D_{\rL(g)} - D_{\rR(g)}\in\mathbb{Z}$, $\widetilde U$ satisfies the condition of Proposition \ref{pr:2dcylinder}.
\item Note that $\widetilde \cH$ can be seen as a subspace
of $\hat \cH_\rL \otimes \hat \cH_\rR := \left(\bigoplus_{g \in G} \cH_\rL^{\rL(g)}\right)\otimes \left(\bigoplus_{g \in G} \cH_\rR^{\rR(g)}\right)$
in a natural way.
We denote $\hat \rL := \bigoplus_{g \in G} \rL(g), \hat \rR := \bigoplus_{g \in G} \rR(g)$. 
\item Let $\psi_\rL^g, \psi_\rR^g$ be as in Section \ref{charged}, with the $V^{g,h}$ chosen as in Proposition \ref{commutingfieldsZZn}.
For each $g\in G$, we introduce the charge operator $\widetilde \psi^g$ as follows
\begin{align}\label{eq:combinegandg}
\widetilde \psi^g := \psi_\rL^g \otimes \psi_\rR^g \in \cB(\hat \cH_\rL) \otimes \cB(\hat \cH_\rR).
\end{align}
It is then clear that $\widetilde \psi^g$ preserves $\widetilde \cH$.
 
 From this definition and by \eqref{eq:chargedintertwiner}, it follows that
 \begin{align}\label{eq:commof}
  \widetilde\psi^g \hat \rL(x_\rL) \otimes \hat \rR(x_\rR)
  = \hat \rL(\rL(g)(x_\rL)) \otimes \hat \rR(\rR(g)(x_\rR)) \widetilde\psi^g,
 \end{align}
 just as in \eqref{eq:chargedfieldchiral}.

 \item The local algebras of $\widetilde\cA$ are given as follows:
 For the diamond $I \times I$, we set $\widetilde\cA(I\times I) := \widetilde\tau(\cA_\rL(I)\otimes \cA_\rR(I)) \vee \{\widetilde\psi^g, (\widetilde\psi^g)^*\}_{g \in G}$.
 For any other diamond $\gamma_\rL \cdot I \times \gamma_\rR \cdot I$,
 we set $\widetilde\cA(\gamma_\rL \cdot I \times \gamma_\rR \cdot I) := \Ad \widetilde U(\gamma_\rL \times \gamma_\rR)(\widetilde\cA(I\times I))$.
\end{itemize}

\begin{theorem}\label{th:2dext}
 The net $(\widetilde\cA, \widetilde U, \widetilde\Omega)$ satisfies (2dCN\ref{2dcn:isotony})--(2dCN\ref{2dcn:vacuum})
 and the condition of Proposition \ref{pr:2dcylinder}, hence provides a conformal net on $\cE$, extending the two-dimensional conformal net $\cA_\rL \otimes \cA_\rR$.
\end{theorem}
\begin{proof}
 Most of the properties, well-definedness of $\widetilde\cA$, (2dCN\ref{2dcn:isotony})--(2dCN\ref{2dcn:vacuum}) except for (2dCN\ref{2dcn:locality})
 and the condition of Proposition \ref{pr:2dcylinder} can be verified as in \cite[Section 6]{MTW18},
 using the assumption $D_{\rL(g)} - D_{\rR(g)}\in\mathbb{Z}$.
 As for locality, in the proof of \cite[Section 6]{MTW18},
 it is only important that $\overline{\epsilon^{\pm}_{\rL(g_1), \rL(g_2)}} = \epsilon^{\mp}_{\rR(g_1), \rR(g_2)}$
 as scalars, which we assume.
 The commutation between $(\widetilde\psi^{g_1})^*$ and $\Ad U(\gamma_\rL\times\gamma_\rR)(\widetilde\psi^{g_2})$ follows
 from $\overline{\epsilon^{\pm}_{\rL(g_1)^{-1}, \rL(g_2)}} = \epsilon^{\mp}_{\rR(g_1)^{-1}, \rR(g_2)}$.
 Therefore, the proof of locality (2dCN\ref{2dcn:locality}) works as well.
\end{proof}

\subsection{Comparison with Longo--Rehren extensions}\label{LRsection}

In this section, we review the so-called Longo--Rehren extensions introduced in \cite[Proposition 4.10]{LR95}, see also \cite{Masuda00}, \cite[Appendix A]{KLM01}, \cite[Section 3.3]{BKLR15} and references therein, and we interpret them as \lq\lq generalized shift constructions". Later, we shall relate them to the extensions presented in the previous sections. 

Let $\cN$ be an infinite factor (\eg, $\cN = \cA_\kappa(I)$ a local algebra of a conformal net on $S^1$). 
Let $\cC$ be a \textbf{unitary fusion category} (not necessarily braided, for the moment) realized as a full subcategory of $\mathrm{End}(\cN)$, the set of normal injective unital *-endomorphisms of $\cN$. 
For the preliminaries on unitary fusion categories (realized, without loss of generality, as endomorphisms of von Neumann algebras) we refer, \eg, to \cite{BKLR15}, \cite{EGNO15}.
We use the notation $\Hom_\cC(\rho,\sigma)$ for the spaces of intertwiners $V\in\cN$ between $\rho$ and $\sigma$ objects in $\cC$, namely $V \rho(x) = \sigma(x) V$ for every $x\in\cN$, the arrows in our subcategory of endomorphism.

For every unitary equivalence class of irreducible objects in $\cC$, choose one representative $\rho_i$, $i=0,\ldots, n$, with $\rho_0 = \id$, and denote $\Delta := \{\rho_0, \ldots, \rho_n\}$. Note that, if $\cC, \cD \subset \mathrm{End}(\cN)$ are two unitary fusion categories, then the Deligne tensor product $\cC \boxtimes \cD$ \cite[Definition 1.1.15]{BaKi01} can be realized in $\mathrm{End}(\cN \bar\otimes \cN)$, where $\cN \bar\otimes \cN$ is the spatial tensor product\footnote{In this section, we denote by $\bar \otimes$ the spatial tensor product of operators, in order to distinguish it from the categorical tensor product functor $\otimes$, which we shall employ below on intertwiners.} von Neumann algebra. Next, let $\cJ : \cN \to \cN^\cJ$ be an antilinear isomorphism of von Neumann algebras, where $\cN^\cJ := \cJ(\cN)$, \eg, the natural involution $\cJ : x \mapsto x^*$, where $\cN^\cJ \cong \cN^\textrm{opp}$. Consider the unitary fusion category $\cC^\cJ \subset \mathrm{End}(\cN^\cJ)$ with objects $\rho^\cJ := \cJ \circ \rho \circ \cJ^{-1} \in \mathrm{End}(\cN^\cJ)$, where $\rho$ is an object in $\cC$, and with arrows $t^\cJ := \cJ(t)$, where $t$ is an arrow in $\cC$.

The \textbf{Longo--Rehren extension} of the spatial tensor product von Neumann algebra $\cN \bar\otimes \cN^\cJ$ (an irreducible finite index unital inclusion of factors $\cN \bar\otimes \cN^\cJ \subset \cM$) is specified by the irreducible Q-system $(\Theta_{\LR}, W_{\LR}, X_{\LR})$ in $\cC\boxtimes \cC^\cJ \subset\textrm{End}(\cN \bar\otimes \cN^\cJ)$ defined as follows. Let
$$\Theta_{\LR} := \bigoplus_{i=0,\ldots,n} \rho_i \bar \otimes \rho_i^\cJ\in \mathrm{End}(\cN \bar \otimes \cN^\cJ).$$
The direct sum is defined by choosing a family of $n+1$ isometries $T_{i,i}\in\cN \bar \otimes \cN^\cJ$, $i=0,\ldots,n$, \lq\lq mutually orthogonal" in the sense that $T_{i,i}^*T_{i',i'} = \delta_{i,i'} \bb 1$, and \lq\lq complete" in the sense that $\sum_{i=0,\ldots,n} T_{i,i}T_{i,i}^* = \bb 1$, and letting $\Theta_{\LR}(x) := \sum_{i=0,\ldots,n} T_{i,i} \rho_i \bar \otimes \rho_i^\cJ (x) T_{i,i}^*$ for every $x\in\cN\bar \otimes \cN^\cJ$. Hence, by definition, $T_{i,i} \in \Hom_{\cC\boxtimes \cC^\cJ}(\rho_i \bar \otimes \rho_i^\cJ, \Theta_{\LR})$ for every $i=0,\ldots,n$. Note that the endomorphism $\Theta_{\LR}$
is not a ``simple tensor'' in $\cC\boxtimes \cC^\cJ$, unless $n=0$.

Let $W_{\LR} := T_{0,0} \in \Hom_{\cC\boxtimes \cC^\cJ} (\id \bar \otimes \id, \Theta_{\LR})$ be the unit of the Q-system and let
$$X_{\LR} := \sum_{i,j,k = 0,\ldots, n} \sqrt{\frac{d(\rho_i) d(\rho_j)}{d(\rho_k)}} (T_{i,i} \otimes T_{j,j}) (\sum_{V} V \bar \otimes V^\cJ) T_{k,k}^* \in \Hom_{\cC\boxtimes \cC^\cJ} (\Theta_{\LR}, \Theta_{\LR}^2)$$ 
be the comultiplication of the Q-system, where $V \in \cN$ runs in a chosen orthonormal basis of $N^k_{i,j} (= \text{the multiplicity of $\rho_k$ in $\rho_i \rho_j$}$, possibly $0$) isometries of $\Hom_\cC (\rho_k, \rho_i \rho_j)$ (whose dimension indeed equals $N^k_{i,j}$) for every fixed $i,j,k = 0,\ldots, n$, with respect to the inner product $V^* V' = \delta_{V,V'} 1$. The definition of $X_{\LR}$ is \emph{independent} of this choice because $\cJ$ is antilinear. Note that $T_{i,i} \otimes T_{j,j}$ denotes the tensor product of arrows in the category of endomorphisms, 
and that both $X_{\LR}$ and $W_{\LR}$ belong to $\cN \bar \otimes \cN^\cJ$. The normalizations read $W_{\LR}^* W_{\LR} = \bb 1$ and $X_{\LR}^* X_{\LR} = (\sum_i d(\rho_i)^2) \bb 1$. Being $(\Theta_{\LR}, W_{\LR}, X_{\LR})$ a Q-system, see \cite[Proposition 4.10]{LR95} for the proof and \cf \cite[Definition 3.8]{BKLR15}, \cite[Proposition 3.19]{BKLR15}, then by \cite{Longo94} $\Theta_{\LR} \in \mathrm{End}(\cN \bar \otimes \cN^\cJ)$ is a dual canonical endomorphism of an irreducible finite index unital extension $\cN \bar\otimes \cN^\cJ \subset \cM$ realized on the Hilbert space $\cH \bar \otimes \cH^\cJ$ if $\cN$ and $\cN^\cJ$ are respectively realized on $\cH$ and $\cH^\cJ$. The beginning of the Jones/canonical tunnel, see \cite[Section 2.5]{LR95}, reads 
$$\cM_{-2} \subset \cM_{-1} \subset \cN \bar \otimes \cN^\cJ \subset \cM,$$
where $\cM_{-2} := \Theta_{\LR}(\cN \bar \otimes \cN^\cJ)$ is the image of $\Theta_{\LR}$, and $\cM_{-1} := \langle \Theta_{\LR}(\cN \bar \otimes \cN^\cJ), X_{\LR}\rangle$ is the von Neumann algebra generated together with $X_{\LR}$. The two subfactors $\cM_{-2} \subset \cM_{-1}$ and $\cN \bar \otimes \cN^\cJ \subset \cM$ are spatially isomorphic.
The Q-system $(\Theta_{\LR}, W_{\LR}, X_{\LR})$ also specifies a (unique, by irreducibility, normal faithful) conditional expectation $E$ from $\cM$ onto $\cN\bar\otimes\cN^\cJ$ with Jones projection $e \in \cM_1$ (the Jones extension of $\cM$ with respect to $E$), or equivalently $E_{-2}(\cdot) := \Theta_{\LR}(W_{\LR}^* \cdot W_{\LR})$ from $\cM_{-1}$ onto $\cM_{-2}$ with Jones projection $e_{-2} := W_{\LR}W_{\LR}^*$. 

\begin{remark}
Choosing a state $\omega$ on $\cN \bar \otimes \cN^\cJ$ (\eg, the vacuum state in the QFT context) and working in the GNS Hilbert space of $(\cM, \omega \circ E)$ yields the more common description in QFT of the extension $\cN\bar\otimes\cN^\cJ \subset \cM$ with a vacuum vector $\Omega$ (cyclic and separating for $\cM$ and inducing an $E$-invariant state) such that $e = [\cN\bar\otimes\cN^\cJ \Omega]$,
\cf the proof of \cite[Theorem 4.9]{LR95}.
\end{remark}

Let $\psi_{j,j} := \Theta_{\LR}(T_{j,j}^*) X_{\LR}\in\cM_{-1}$, for every $j=0,\ldots,n$. 
Each $\psi_{j,j}$ has the \lq\lq charged field intertwiner" property on $\cM_{-2} \subset \cM_{-1}$ in the sense of Doplicher--Roberts \cite{DR72}, \ie, 
$$\psi_{j,j} \Theta_{\LR}(x) = \Theta_{\LR}(\rho_j \bar \otimes \rho_j^\cJ(x)) \psi_{j,j}, \quad x \in \cN \bar \otimes \cN^\cJ.$$ 
We refer to the $\psi_{j,j}$ as the \textbf{charged fields} of the Longo--Rehren Q-system $(\Theta_{\LR}, W_{\LR}, X_{\LR})$. Such charged fields can be defined for every other Q-system of endomorphisms.

As $X_{\LR} = \sum_{j=0,\ldots,n} \Theta_{\LR}(T_{j,j}) \psi_{j,j}$, because $\sum_{j=0,\ldots,n} T_{j,j} T_{j,j}^* = \bb 1$, it holds
$$\cM_{-1} = \langle \cM_{-2}, \psi_{0,0}, \ldots, \psi_{n,n}\rangle.$$
Note that $\psi_{0,0} = \Theta_{\LR}(T_{0,0}^*) X_{\LR} = \Theta_{\LR}(W_{\LR}^*) X_{\LR} = \bb 1$.

Setting $\cH_{j,j} := T_{j,j}T_{j,j}^* \cH \bar \otimes \cH^\cJ$, for every $j=0,\ldots,n$, yields an orthogonal decomposition 
\begin{align}\label{eq:HilbLRext}
\cH \bar \otimes \cH^\cJ = \bigoplus_{j=0,\ldots,n} \cH_{j,j}.
\end{align}
Note that $\cH_{0,0}$ is the range of the Jones projection $e_{-2} = T_{0,0}T_{0,0}^* = W_{\LR}W_{\LR}^*$.

Moreover, by definition of $X_{\LR}$ and $T_{i,i} \otimes T_{j,j} = \Theta_{\LR}(T_{j,j}) T_{i,i}$, it follows
$$\psi_{j,j} = \sum_{i,k = 0,\ldots, n} \sqrt{\frac{d(\rho_i) d(\rho_j)}{d(\rho_k)}} T_{i,i} (\sum_{V} V \bar \otimes V^\cJ) T_{k,k}^*,$$
where the first sum runs over all $i,k$ such that $\Hom(\rho_k,\rho_i \rho_j) \neq \langle 0\rangle$ for fixed $j$, and the second sum runs over $V$ in the chosen orthonormal basis of isometries in $\Hom_\cC (\rho_k, \rho_i \rho_j)$. Note that, while $\cM_{-2}$ acts diagonally on $\cH \bar \otimes \cH^\cJ$ by the very definition of $\Theta_{\LR}$, the charged fields $\psi_{j,j}$'s mix different components.

\begin{remark}
In the special case of \textbf{pointed} unitary fusion categories $\cC$ (in particular, $d(\rho_i) = d(\rho_j) = d(\rho_k) = 1$ and the fusion ring is a finite group), then all vector spaces $\Hom_\cC (\rho_k, \rho_i \rho_j)$ have either dimension 1 or 0. In this case each $\psi_{j,j}$ operates as a \lq\lq right $j^{-1}$ shift operator" on the grading of $\cH \bar \otimes \cH^\cJ$, mapping each subspace $\cH_{k,k}$, $k=0,\ldots,n$, to the subspace $\cH_{i,i}$ such that $[\rho_k] = [\rho_i \rho_j]$, \ie, $[\rho_k \rho_j^{-1}] = [\rho_i]$. Moreover, for every $j,j' = 0,\ldots,n$, 
\begin{align*}
\psi_{j,j} \psi_{j',j'} = \sum_{i,k = 0,\ldots, n} T_{i,i} (V \bar \otimes V^\cJ) T_{k,k}^* \sum_{i',k' = 0,\ldots, n} T_{i',i'} (V' \bar \otimes V'^\cJ) T_{k',k'}^* 
\end{align*} 
where $i,k,i',k'$ are such that $\dim(\Hom_\cC(\rho_k,\rho_i \rho_j))=1$ and $\dim(\Hom_\cC(\rho_{k'},\rho_{i'} \rho_{j'}))=1$. Note that, as $\Delta$ is fixed,  $V,V'$ are both unique up to a phase factor, and $\cJ$ is antilinear, hence the irrelevance of their choice is immediately evident in the pointed case. If we could choose $\Delta$ to be a group (\ie, closed under multiplication, as $\rho_0 = \id$), then using $T_{k,k}^* T_{i',i'} = \delta_{k,i'} \bb 1$ and $VV' \in \Hom(\rho_{k'}, \rho_i \rho_j \rho_{j'}) = \Hom(\rho_{k'}, \rho_i \rho_h)$ for some $\rho_h\in\Delta$ such that $\rho_j \rho_{j'} = \rho_h$, then it would follow that $\psi_{j,j} \psi_{j',j'} = \psi_{h,h}$. If we chose the inverse (which is a conjugate in the case of automorphisms) of some $\rho_j$ to be some $\rho_l \in \Delta$, together with $\rho_0 = \id$, then $\psi_{j, j}^* = \psi_{l, l}$. 
\end{remark}

If $\cC$ is in addition \textbf{braided}, not necessarily pointed, and $\cC^\cJ$ is endowed with the braiding induced by $\cJ$, then the Longo--Rehren Q-system is \emph{commutative}. See the discussion following the proof of \cite[Proposition 4.10]{LR95} and \cf \cite[Proposition 4.21]{BKLR15}. Let $\epsilon^{\pm}_{\rho_i,\rho_j}$ be the braiding in $\cC$ and its opposite, then the commutativity condition reads $\epsilon^{\pm}_{\Theta_{\LR}, \Theta_{\LR}} X_{\LR} = X_{\LR}$, \ie, for every $i,j = 0,\ldots, n$,
\begin{align}\label{eq:commutLRchargedfields}
\epsilon^{\pm}_{\rho_i,\rho_j} \bar\otimes\, {\epsilon^{\pm, \cJ}_{\rho_i,\rho_j}}\, \psi_{i,i} \psi_{j,j} = \psi_{j,j} \psi_{i,i},
\end{align}
where $\epsilon^{\pm}_{\rho_i,\rho_j} \bar\otimes\, {\epsilon^{\pm, \cJ}_{\rho_i,\rho_j}}$ is the braiding between $\rho_i \bar \otimes \rho_i^\cJ$ and $\rho_j \bar \otimes \rho_j^\cJ$ in $\cC \boxtimes \cC^\cJ$, the irreducible summands of $\Theta_{\LR}$.

\begin{remark}
If $\cC$ is braided and pointed, then each braiding $\epsilon^{\pm}_{\rho_i,\rho_j}$ is a phase multiple of any fixed unitary intertwiner $U \in \Hom_\cC(\rho_i \rho_j,\rho_j \rho_i)$ (both $\rho_i \rho_j$ and $\rho_j \rho_i$ are automorphisms hence irreducible). Thus $\epsilon^{\pm}_{\rho_i,\rho_i} \bar\otimes\, {\epsilon^{\pm, \cJ}_{\rho_i,\rho_i}} = \bb 1 \bar \otimes \bb 1$ for $i=j$, by taking $U = \bb 1$ and by the antilinearity of $\cJ$. Equivalently, the statistical phase of each $\rho_i \bar \otimes \rho_i^\cJ$ is $+1$ (all bosonic automorphisms in $\cC \boxtimes \cC^\cJ$) being equal to the (phase associated with the) self-braiding in the case of automorphisms. If, in addition, $\rho_i \rho_j = \rho_j \rho_i$ for every $i \neq j$ as automorphisms of $\cN$ (\eg, by spacelike localization in the case of DHR automorphisms and $\cN = \cA_\kappa(I)$, or if $\rho_i$ and $\rho_j$ are different powers of the same automorphism as for cyclic groups), then one can take $U = \bb 1$ as well, and $\epsilon^{\pm}_{\rho_i,\rho_j} \bar\otimes\, {\epsilon^{\pm, \cJ}_{\rho_i,\rho_j}} = \bb 1 \bar \otimes \bb 1$ for every $i, j = 0,\ldots,n$. In the latter case of commuting automorphisms, the commutativity of the Longo--Rehren Q-system for braided pointed fusion categories reads
$$\psi_{i,i} \psi_{j,j} = \psi_{j,j} \psi_{i,i}.$$
\end{remark}

If we take $\cC$ to be the representation category of a (completely) rational conformal net on $\bbR$ or $S^1$, then $\cC$ is necessarily unitary fusion (finitely many sectors) and braided (also modular, see, \eg, \cite{BaKi01}), equipped with the DHR braiding, by a result of \cite{KLM01}. Then the Longo--Rehren Q-system provides an irreducible finite index \emph{local} (by the commutativity property \eqref{eq:commutLRchargedfields}) and \lq\lq diagonal" extension $\tilde \cA_{\LR}$ of the tensor product conformal net $\cA_\kappa \bar \otimes \cA_\kappa^\cJ$ on $\bbR^{1+1}$ (with equal left and right chiral components), by choosing $\cN \bar \otimes \cN^\cJ = \cA_\kappa(I) \bar \otimes \cA_\kappa(I)^\cJ$, where $I$ is an interval on $\bbR$ symmetric around the origin, and $\cJ$ the modular conjugation with respect to $\cA_\kappa(\bbR_+), \Omega_\kappa$ (with this choice, $\rho_i^\cJ \cong \bar \rho_i$, hence in particular $\rho_i^\cJ \cong \rho_i^{-1}$ in the pointed case).
See \cite[Theorem 4.9]{LR95}, \cf \cite[Theorem 6.8]{DG18}, for the general construction of extensions from arbitrary Q-systems in the representation category of a local net, and the comments after the proof of \cite[Proposition 4.10]{LR95} for the Longo--Rehren Q-system.

\section{Conformal Wightman fields from charged primary fields}

\subsection{Conformal Wightman axioms}

Here we show that to some of the extensions we discussed in Sections \ref{net-trivialtotalbraiding} and \ref{2dnet}
we can associate conformal Wightman fields, as we define below.
It is natural to expect that such Wightman fields on $\bbR^{1+1}$ also extend to $\cE$ as conformal nets do,
and hence we formulate Wightman fields on $\cE$.

For simplicity, let us start with chiral fields on $S^1$.
Let us see $S^1 \subset \bbC$ and the counterclockwise direction as the lightlike direction
when $S^1$ is seen as the one-point compactification of $\bbR$ as in Section \ref{cylinder}.

A conformal Wightman field on $S^1$ with \textbf{conformal dimension} $D$
is a quantum field that transforms as the tensor field $f\frac{d}{d\theta}^{\otimes D-1}$ under diffeomorphisms
as (1dW\ref{1dw:diff}) below,
and such a field is called a \textbf{primary field}:
For $D \in \bbN, \gamma \in \diff(S^1)$ and $f \in C^\infty(S^1)$, we set (\cf \cite{RTT22})
\begin{align*}\label{eq:beta}
X_\gamma(e^{i\theta}) &:= -i \frac{d}{d \theta} \log(\gamma(e^{i \theta})), \\
(\beta_D(\gamma)f)(z) &:= (X_\gamma(\gamma^{-1}(z)))^{D-1} f(\gamma^{-1}(z)).
\end{align*}
Note that $\gamma$ is orientation-preserving, hence $X$ is strictly positive
and $(X_\gamma(\gamma^{-1}(z)))^{D-1}$ is bounded below by some positive number for a fixed $\gamma$.

A \textbf{conformal Wightman field theory on $S^1$} is a family of operator-valued distributions $\{\phi_{\kappa,j}\}$
($\kappa$ is a fixed label, while $j$ indexes the family)
on $S^1$,
closed under the conjugate ($\phi_{\kappa,j}^\dagger$ is also in the family),
with corresponding conformal dimensions $\{D_{\kappa,j}\}$ defined on a common invariant dense domain $\scD \subset \cH_\kappa$,
a unitary projective representation $U_\kappa$ of $\diff(S^1)$ and a vector $\Omega_\kappa \in \cH_\kappa$ such that
\begin{enumerate}[{(1dW}1{)}]
 \item \textbf{Locality:} for $f, g \in C^\infty(S^1)$ with disjoint supports, $[\phi_{\kappa, j_1}(f), \phi_{\kappa, j_2}(g)] = 0$ on $\scD$. \label{1dw:locality}
 \item \textbf{Diffeomorphism covariance:} \label{1dw:diff} $U_\kappa(\gamma)$ preserves $\scD$ and it holds that
 \[
  \Ad U_\kappa(\gamma)(\phi_{\kappa,j}(f)) = \phi_{\kappa,j}(\beta_{D_{\kappa,j}}(\gamma)f), \qquad \text{ for } \gamma \in\diff(S^1).
 \]
 \item\textbf{Positivity of energy:}
 the spectrum of rotations of $U_\kappa$ is contained in $\bbN \cup \{0\}$. \label{1dw:positiveenergy}
 \item \textbf{Vacuum and the Reeh-Schlieder property: }there exists a unique (up to a phase)
 vector $\Omega_\kappa \in\cH_\kappa$ such that $U_\kappa(\gamma)\Omega_\kappa=\Omega_\kappa$ for $\gamma\in \Mob$
 and vectors of the form $\phi_{\kappa, j_1}(f_1)\cdots \phi_{\kappa, j_\ell}(f_\ell)\Omega_\kappa$,
 where $f_1, \cdots, f_\ell \in C^\infty(S^1)$, are total in $\cH_\kappa$. \label{1dw:vacuum}
 \setcounter{1dW}{\value{enumi}}
\end{enumerate}

As fields $\phi_{\kappa,j}$ are operator-valued distributions, they can be smeared with
functions $\mathfrak{e}_n(\theta) = e^{in\theta}$ to give their Fourier components $\phi_{\kappa, j,n} = \phi_{\kappa,j}(\mathfrak{e}_n)$.
Let $L_{\kappa,0}$ be the generator of $U_\kappa(R_t)$, where $R_t \in \Mob$ are rotations.
We assume that $\scD = C^\infty(L_{\kappa,0}) := \bigcap_{\ell\in \bbN}\dom(L_{\kappa,0}^\ell)$.
Furthermore, 
\begin{enumerate}[{(1dW}1{)}]
 \setcounter{enumi}{\value{1dW}}
 \item \label{1dw:peb}
 \textbf{Polynomial energy bounds:} there are $r_j, p_j, C > 0$ such that,  for $\Psi \in C^\infty(L_{\kappa, 0})$,
 \[
 \|\phi_{\kappa, j,n}\Psi\| \le C(1+|n|)^{r_j} \|(L_{\kappa, 0} + \bb1)^{p_j}\Psi\|
 \]
 \setcounter{1dW}{\value{enumi}}
\end{enumerate}

According to \cite[Section 6]{CKLW18}, polynomial energy bounds allow one to define the smeared fields.
For $f \in C^\infty(S^1)$ whose Fourier components are $f_n = \frac1{2\pi}\int_{-\pi}^{\pi} f(e^{-in\theta})d\theta$, we define
\begin{align*}
 \phi_{\kappa, j}(f) := \sum_{n \in \bbZ} f_n \phi_{\kappa, j,n}.
\end{align*}
This is convergent on any vector $\Psi \in C^\infty(L_{\kappa,0})$ and $\phi_{\kappa, j}(f)$ preserves the domain $C^\infty(L_{\kappa,0})$.

We know from (the proofs of) \cite[Proposition 6.4, Theorem 8.3]{CKLW18} the following:
\begin{lemma}
 Let $\{\phi_{\kappa,j}\}$ satisfy polynomial energy bounds with $p_j\le1$.
 Then it holds that $\Ad U(\gamma)(\phi_{\kappa,j}(f)) = \phi_{\kappa,j}(\beta_{D_{\kappa,j}}(\gamma)(f))$ and
 $\phi_{\kappa,j_1}(f)$ and $\phi_{\kappa, j_2}(g)$ commute strongly for $f,g$ with disjoint supports.
\end{lemma}

A general two-dimensional conformal field is not chiral and depends on both lightlike variables.
There are fields called primary fields and they are distinguished by the conformal dimensions $(D_\rL, D_\rR)$,
where $D_\rL, D_\rR > 0$ and can be non-integer. We set $\beta_{D_\rL, D_\rR}(\gamma) = \beta_{D_\rL}(\gamma_\rL)\beta_{D_\rR}(\gamma_\rR)$
for $\gamma = \gamma_\rL \times \gamma_\rR \in \overline{\diff(S^1)}\times\overline{\diff(S^1)}$.
This can be extended to $\scC$ if $D_\rL - D_\rR \in \bbZ$, which we always assume (cf. Section 4.1).
Let $H$ be the generator of $\{R_t \times R_t: t \in \bbR\}$ in $U$,
which plays a similar role to that of $L_0$ in chiral fields.

A \textbf{conformal Wightman field theory on $\cE$} is a family of operator-valued distributions $\{\widetilde\psi_j\}$ on $\cE$,
closed under conjugation, with corresponding conformal dimensions $\{(D_{\rL,j}, D_{\rR, j})\}$ defined on a common invariant domain $\scD \subset \cH$
consisting of vectors of the form $\widetilde\psi_{j_1}(f_1)\cdots \widetilde\psi_{j_n}(f_n)\Omega$,
a unitary projective representation $U$ of $\scC$ and a vector $\Omega \in \cH$ such that
\begin{enumerate}[{(2dW}1{)}]
 \item \label{2dw:locality}\textbf{Locality:} for $f, g \in C^\infty(\cE)$ with spacelike separated supports, $[\widetilde\psi_{j_1}(f), \widetilde\psi_{j_2}(g)] = 0$.
 \item \label{2dw:poincare}\textbf{Diffeomorphism covariance:} $U(\gamma)$ preserves $\scD$ and it holds that
 \[
  \Ad U(\gamma)(\widetilde\psi_j(f)) = \widetilde\psi_j(\beta_{D_{\rL, j}, D_{\rR, j}}(\gamma)f), \qquad \text{ for } \gamma\in\scC
 \]
 \item \label{2dw:positiveenergy}\textbf{Positivity of energy:}
 the joint spectrum of the translation subgroup of $\bbR^{1+1}$ in $U$ is contained in the closed forward light cone
 $\overline {V_+}=\{(a_0,a_1)\in\bbR^{1+1}: a_0^2-a_1^2\geq0, a_0\geq0\}$.
 \item \label{2dw:vacuum} \textbf{Vacuum and the Reeh-Schlieder property: }there exists a unique (up to a phase)
 vector $\Omega \in\cH$ such that $U(\gamma)\Omega=\Omega$ for $\gamma\in \overline{\Mob}\times\overline{\Mob}/\mathfrak R$
 and $\cH$ is spanned by vectors of the form $\widetilde\psi_{j_1}(f_1)\cdots \widetilde\psi_{j_k}(f_k)\Omega$.
 \item \label{2dw:peb}
 \textbf{Polynomial energy bounds:}
 it holds that $\scD = C^\infty(H)$ and there are $p_j > 0$ such that, for $\Psi \in C^\infty(H)$,
 \[
 \|\psi_j(f)\Psi\| \le C_f \|H^{p_j}\Psi\|,
 \]
 for some $C_f>0$ depending only on $f$.
 \setcounter{2dW}{\value{enumi}}
\end{enumerate}

As is well-known, linear energy bounds ($p_j = 1$) assure that the conformal Wightman fields commute strongly.
\begin{lemma}
 Let $\{\widetilde\psi_j\}$ satisfy the bound $\|\widetilde\psi_j(f)\Psi\| \le C_{f,j}\|H\Psi\|$.
 Then $\widetilde\psi_{j_1}(f)$ and $\widetilde\psi_{j_2}(g)$ commute strongly for $f,g$ with spacelike supports.
\end{lemma}
\begin{proof}
 As $H$ is the generator of the one-parameter group $R_t\times R_t$ in $\overline{\Mob}\times\overline{\Mob}$,
 we have $[H, \widetilde\psi_j(f)] = i\widetilde\psi_j(f')$ as in \cite[(3.7)]{CDVIT21},
 where $f'$ is the derivative of $f$ with respect to translations $R_t\times R_t$ of the cylinder $\cE$.
 Furthermore, we have $\|[H, \widetilde\psi_j(f)]\Psi\| = \|\widetilde\psi_j(f')\Psi\| \le C_{f',j}\|H\Psi\|$.
 By applying commutator with $H$, $\|\delta^k(\widetilde\psi_j(f))\Psi\| = \|\widetilde\psi_j(f^{(k)})\Psi\| \le C_{f^{(k)},j}\|H\Psi\|$, $k=2,3$,
 where $\delta(A) = i[H, A]$.
 Therefore, $\widetilde\psi_1(f)$ and $\widetilde\psi_2(g)$ satisfy the assumption of Theorem \ref{th:locality-doublecomm} with $H$ as the reference operator
 and they strongly commute.
\end{proof}

\subsection{Formal series of operators}\label{formal}

Let $A_s$ be a family of operators parametrized by $s \in \bbR$.
By a formal series we mean a symbol of the form $\sum_s A_s z^s$,
where the summation actually has no meaning.
We refer to $A_s$ as \textbf{Fourier components} of a formal series.
One can consider sums and scalar multiples of such formal series:
$\sum_s A_s z^s + \sum_s B_s z^s = \sum_s (A_s+B_s) z^s$,
$c\sum_s A_s z^s = \sum_s cA_s z^s$.
The product of two formal series is not always defined, but in some cases
there is a natural product. We define the product of $\sum_s A_s z^s$
and $\sum_s B_s z^s$ as the formal series $\sum_s C_s z^s$ with 
the coefficient $C_s = \sum_{t\in \bbR} A_{s-t}B_t$,
whenever these sums make sense.

Similarly, we consider formal series in two or more variables.
If we have two formal series in two different variables $z,w$,
the product $\sum_s A_s z^s \sum_t B_t w^t = \sum_{s,t} A_s B_t z^s w^t$
makes always sense. A product of such formal series 
in two variables is defined similarly as above.

A typical case we use in this paper is when the vector spaces are graded as $V^{j} = \bigoplus_{t \in \bbN + h_j} V^j_t$,
where $j=1,2$, $h_j \in \bbR$ and $A_s$ are operators $V^1_t\to V^2_{t-s}$ and
they are nonzero only for $s \in \bbZ + D$ for some $D\in \bbR$ and $t-s \in \bbN + h_2$.
We also consider the situation where $V = \bigoplus_j V^j$
and $\sum_s A_s$ is a formal series on $V$, that has the form for each pair $V^{j_1}, V^{j_2}$ as above
(and $D_{j_1, j_2}$ may depend on $j_1, j_2$).

Let $\sum_s A_s z^s, \sum_s B_s z^s$ such two formal series with $[A_{s_1}, B_{s_2}] = 0$
for all $s_1, s_2 \in \bbR$. In the situation of the previous paragraph, we can define the product
$\sum_s A_s z^s \sum_t B_s z^s = \sum_{s}z^s\sum_{t} A_{s-t}B_t$
(this is a special case of normal product, see \cite[Section 2.1]{CTW23}):
indeed, for each vector $\Psi \in V$ the sum over $t$ is finite because
either $A_{s-t}\Psi$ or $B_t\Psi$ vanishes for large $|t|$. Therefore, the sum $\sum_{t} A_{s-t}B_t$
defines an operator on $V$ and this is a formal series on $V$.

\subsection{Charged primary fields}\label{charged-fields}

We continue studying a single chiral component, but we omit $\kappa$ from certain parameters that do not appear later. We assume that a conformal net $(\cA_\kappa, U_\kappa, \Omega_\kappa)$ on $S^1$ is generated by conformal Wightman fields $\{\phi_{\kappa,j}\}$ which commute strongly when smeared with test functions with disjoint supports. 

Unitary operators $e^{i\phi_{\kappa,j}(f)}$ with compactly supported $f$ are represented in a representation $\rho$ as a unitary operator, and the question arises whether the operator $\rho(\phi_{\kappa,j}(f))$ makes sense and whether it is the generator of $\rho(e^{i\phi_{\kappa,j}(f)})$. This property is called \textit{strong integrability}, and veryfing it in examples is a non-trivial problem (although it holds in interesting cases, like the U(1)-current \cite{CWX, Gui20-2}). Furthermore, we can also consider localized unbounded intertwiners between representations and there is the problem whether the commutation relations between them hold strongly (\textit{strong intertwining property} and \textit{strong braiding}) \cite{Gui20-2},
see also \cite{Tener19Representation, Tener19Fusion}.
These properties should be useful in comparing the extension of nets and Wightman fields, see Remark \ref{rm:strong}.

Let $\kappa(g)$ be an automorphism of the conformal net $\cA_\kappa$ for $g \in G$ as in Section \ref{chiral}.
We consider the Fourier components of chiral fields and assume that they have corresponding operators $\phi_{\kappa,j,n}^{\kappa(g)}$
on the representation spaces $\cH_{\kappa}^{\kappa(g),\fin}$, where $\cV^\fin$ denotes the linear span of eigenspaces
$\cV_t = \ker(L_{\kappa,0} - t)$ of $L_{\kappa,0}$,
where $L_{\kappa,0}$ is represented on $\cV$ as an unbounded operator.
We denote the algebraic direct sum of the $\phi^{\kappa(g)}_{\kappa,j,n}$ over $g\in G$
by $\hat\phi_{\kappa,j,n}$ defined on the (algebraic) linear span $\bigoplus_{g \in G}^\text{alg} \cH_{\kappa}^{\kappa(g),\fin}$
in $\bigoplus_{g \in G} \cH_{\kappa}^{\kappa(g),\fin}$.
Similarly, we denote $\hat L_{\kappa,n} = \bigoplus_{g\in G} L_{\kappa,n}^{\kappa(g)}$,
where $L_{\kappa, n}^{\kappa(g)}$ is a representation of the Virasoro algebra on $\cH_{\kappa}^{\kappa(g), \fin}$.

Let $(\cA_\kappa, U_\kappa, \Omega_\kappa)$ be a conformal net on $S^1$ generated by fields $\{\phi_{\kappa,j}\}$ with conformal dimensions $\{D_j\}$
as above and a collection $\Delta_\kappa$ of automorphisms to which there is a bijective map $\kappa$ from a group $G$.
A \textbf{charged primary field} $\{\psi_\kappa^h(\xi,s)\}_{s \in \bbR}$ is
a family of operators
(labelled by $h \in G$ and $\xi$ in some index set\footnote{Examples where
we need all these labels are the WZW models associated with a simply laced, simply connected, simple, connected compact group $G$ (different from the group above) at level $1$.
The sectors are parametrized by the center $Z(G)$,
which can be identified with $\Lambda_W/\Lambda_R$,
the weight lattice quotiented by the root lattice \cite[Lemma I.2.1.1]{Toledano-LaredoThesis}.
The primary fields are of the form $\Phi_\mu(z)$, where $\mu$ is a minimal weight of $G$ \cite[Section V.5.3]{Toledano-LaredoThesis}.
In this example, $h$ is $\mu + \Lambda_R \in \Lambda_W/\Lambda_R$, while $\mu$ is one of minimal weights,
playing the role of $\xi$ in our notations. See \cite[Theorem V.5.3.2]{Toledano-LaredoThesis}
for relative locality of these fields.}
$\Xi$ depending on $\kappa$ and $h$), acting on the domain $\bigoplus_{g \in G}^\text{alg} \cH_{\kappa}^{\kappa(g),\fin}$
with the associated formal series $\psi_\kappa^h(\xi,z) = \sum_{s \in \bbR} \psi_\kappa^h(\xi,s)z^{-s-D_h}$, where $D_h > 0$ is the conformal dimension and $z$ is the formal variable, such that
\begin{itemize}
 \item $\psi_\kappa^h(\xi, s)$ maps $\cH_\kappa^{\kappa(g),\fin}$ to $\cH_\kappa^{\kappa(hg),\fin}$,
 and there is $\ell_{h,g} \in \bbR$ such that $\psi_\kappa^h(\xi, s) \neq 0$ only for $s \in \ell_{h,g} + \bbZ$.
 \item (primarity) $[\hat L_{\kappa,m}, \psi_\kappa^h(\xi, s)] = ((D_h-1)m - s)\psi_\kappa^h(\xi, m+s)$.
 \item (relative locality) $[\hat \phi_{\kappa,j,m}, \psi_\kappa^h(\xi, s)] = \sum_{\xi'} X_{j,\xi'}^\xi\psi_\kappa^h(\xi',m+s)$, where $X_{j,\xi'}^\xi \in \bbC$.
 \item As formal series, the braiding relation holds:
 \begin{align*}
 (1-\textstyle{\frac zw})^{\alpha(h_1, h_2)}w^{\alpha(h_1, h_2)}\psi_\kappa^{h_1}(\xi_1, w)\psi_\kappa^{h_2}(\xi_2, z)
 = (1-\textstyle{\frac wz})^{\alpha(h_1, h_2)}z^{\alpha(h_1, h_2)}\psi_\kappa^{h_2}(\xi_1, z)\psi_\kappa^{h_1}(\xi_2, w)  
 \end{align*}
 where $\alpha(h_1, h_2) \in \bbR$. Recall that
 $(1-u)^\beta = \sum_{n \ge 0} \binom{\beta}{n}(-u)^n$ as a formal series, 
 where $\binom{\beta}{n} = \frac{\beta(\beta-1)\cdots(\beta-n+1)}{n!}$.
 \item We assume that the family of charged primary fields is closed under conjugation,
 in the sense that for given $h, \xi$ there are $\bar h, \bar \xi$ such that
 $\psi^h_\kappa(\xi, z)^* = \psi^{\bar h}_\kappa(\bar \xi, z)$, with the convention $z^* = z^{-1}$ in the formal series expansion.
\end{itemize}

We further assume polynomial energy bounds:
\begin{itemize}
 \item for each $h$, there are $r_\xi, p_\xi, C_\xi \ge 0$ such that $\|\psi_\kappa^h(\xi,s)\Psi\| \le C_\xi(1+|s|)^{r_\xi} \|(L_{\kappa,0} + \bb1)^{p_\xi}\Psi\|$
 for $\Psi \in C^\infty(L_{\kappa,0})$.
\end{itemize}
This allows us to define smeared charged primary fields:
For $f \in C^\infty(S^1\setminus \{-1\})$, we put $f_s = \frac1{2\pi}\int_{-\pi}^\pi f(e^{i\theta})e^{-is\theta}d\theta$
and
\begin{align*}
 \psi_\kappa^{h}(\xi,f) := \sum_{s \in \bbR} f_s \psi_\kappa^h(\xi,s),
\end{align*}
and this defines operators on $C^\infty(\hat L_{\kappa,0})$
(the sum makes sense because $\psi_\kappa^h(\xi,s) = 0$ except for countable $s$ on each $\cH_{\kappa, t}^{\kappa(g)}$).

 Let\footnote{We use the same notation for the braiding of the charged fields and the braiding of sectors in Section \ref{chiral}.
 We plan to show that they are indeed equal under assumptions such as strong braiding \cite{Gui20-2} in a separate publication.}
 $\epsilon^+_{\kappa(h_1), \kappa(h_2)} := \lim_{\Im \zeta > 0, \zeta\to -1}\zeta^{\alpha(h_1, h_2)}$,
 where $\Im \zeta$ is the imaginary part of $\zeta$, $w^{\alpha(h_1, h_2)} = e^{\alpha(h_1, h_2)\log w}$ and
 we take the branch of $\log w$ on $\bbC \setminus (-\infty,0]$ such that $\log w \in \bbR$ on $(0,\infty)$.
\begin{lemma}\label{lm:braiding}
 Let $f, g \in C^\infty(S^1\setminus \{-1\})$ such that
 $\arg \supp f < \arg \supp g$, that is, they are in the counterclockwise order on $S^1\setminus \{-1\}$.
Then it holds that $\psi_\kappa^{h_1}(\xi_1, f)\psi_\kappa^{h_2}(\xi_2, g) = \epsilon^+_{\kappa(h_1), \kappa(h_2)}\psi_\kappa^{h_2}(\xi_2, g)\psi_\kappa^{h_1}(\xi_1, f)$. 
\end{lemma}
\begin{proof}
 Let $z_0 \in S^1, z_0 \neq -1$ such that $\arg \supp f < \arg z_0 < \arg\supp g$,
 and $I_+$ be the interval on $S^1$ from $z_0$ to $-1$ (counterclockwise)
 while $I_-$ be the interval on $S^1$ from $-1$ to $z_0$, hence we have
 $\supp f \subset I_-, \supp g \subset I_+$  (see Figure \ref{fig:braiding}),
 and put $f_s = \frac1{2\pi}\int_{-\pi}^\pi f(e^{i\theta})e^{-is\theta}d\theta$,
 $g_t = \frac1{2\pi}\int_{-\pi}^\pi g(e^{i\theta})e^{-it\theta}d\theta$.
 As $f,g$ are smooth, these coefficients are rapidly decaying.

 The product $\psi_\kappa^{h_1}(\xi_1, w)\psi_\kappa^{h_2}(\xi_2, z)$ is a formal series
 in $w, z$. For finite-energy vectors $\Psi_1, \Psi_2$ in some $\cH_{\kappa}^{\kappa(h_3),\fin}, \cH_\kappa^{\kappa(h_4),\fin}$,
 the scalar product
 $\varphi_1(w,z) := \<\Psi_1, \psi_\kappa^{h_1}(\xi_1, w)\psi_\kappa^{h_2}(\xi_2, z)\Psi_2\>$ can be considered as a formal series in $w,z$ with
 coefficients in $\bbC$, and by the first assumption on charged primary fields, it is just a countable sum. 
 Furthermore, these coefficients of $w^s z^t$ vanish for $s$ large and $t$ small,
 because each $\cH_{\kappa}^{\kappa(h_j),\fin}$ is a positive-energy representation of the Virasoro algebra.
 By polynomial energy bounds, it can be seen as a distribution in $w,z$,
 which we denote again by $\varphi_1(w,z)$.
 Then it holds that $\varphi_1(f,g) = \sum \varphi_1(s,t)f_s g_t$,
 where $\varphi_1(s,t)$ is the coefficient of the formal series $\varphi_1(w,z)$ of $w^s z^t$
 (with a slight abuse of notation).

 Let $a<1$ and $b>1$. 
 We introduce
 $\varphi_1^{a,b}(z,w) := \<\Psi_1, \psi_\kappa^{h_1}(\xi_1, bw)\psi_\kappa^{h_2}(\xi_2, az)\Psi_2\>$,
 then it has only finitely many terms with negative powers in $z$ and those with positive powers in $w$,
 therefore, it is a convergent series for $|z| = |w| = 1$ again by the polynomial energy bounds
 and the choice $a<1, b>1$.
 Now it can be seen as a function on $I_-\times I_+$, thus defines a distribution and $\varphi_1^{a,b}(f,g) = \sum \varphi^{a,b}_1(s,t)f_s g_t$,
 in the sense above. Furthermore, as $a \to 1, b \to 1$, this converges to $\varphi_1(f,g)$.
 
 The formal series $(1-\textstyle{\frac zw})^{\alpha(h_1, h_2)}$
 has only negative powers in $w$ and positive powers in $z$, therefore, the product
 \begin{align*}
  \<\Psi_1, (1-\textstyle{\frac zw})^{\alpha(h_1, h_2)}\psi_\kappa^{h_1}(\xi_1, w)\psi_\kappa^{h_2}(\xi_2, z)\Psi_2\>
 \end{align*}
 makes sense again as a formal series (see the remark on formal series above).
 Again by polynomial energy bounds, these coefficients grow only polynomially in $s,t$.
 Therefore, it can be seen as a distribution with two variables $w \in I_-, z \in I_+$,
 which we denote by $\varphi_2(w,z)$. Define also $\varphi^{a,b}_2(w,z) = \varphi_2(bw,az)$.
 As $a\to 1, b \to 1$, each coefficients converge, and $\varphi^{a,b}_2(f,g)$ converge
 as distributions to $\varphi_2(f,g)$ by polynomial energy bounds.

 Note that the expansion $(1-u)^\beta = \sum_{n \ge 0} \binom{\beta}{n}(-u)^n$
 converges for $|u| < 1$, and can be analytically continued to $\bbC \setminus [1,\infty)$.
 This expansion coincides with the definition $(1-u)^\beta = e^{\beta \log(1-u)}$
 as an analytic function in $u \in \bbC \setminus [1,\infty)$,
 where $\log w$ is defined on $\bbC \setminus (-\infty,0]$ and $\log w \in \bbR$ when $w \in (0,\infty)$
 (the last condition fixes a branch of $\log w$ uniquely).
 We use the same continuation for $\log w$ when $w^{\beta} = e^{\beta \log w}$ is considered as a function.
 If we consider $(1-\frac w z)^{\alpha(h_1, h_2)}$ as a function in $w,z$,
 we can multiply it to the above distribution $\varphi_1$ and obtain
 a distribution in $I_- \times I_+$, which we denote by $\varphi_3(w,z)$.
 We also introduce $\varphi_3^{a,b}(w,z) = \varphi_3(bw,az)$.
 Then $\varphi_3^{a,b}(f,g)$ converges to $\varphi_3(f,g)$ as $a\to 1, b \to 1$
 because both $\varphi_1^{a,b}(w,z)$ and $(1-\frac{bw}{az})^{\alpha(h_1, h_2)}$
 converge as distributions and smooth functions on $I_-\times I_+$, respectively.
 
 This shows that $\varphi_2(f,g) = \varphi_3(f,g)$.
 That is, $(1-\textstyle{\frac zw})^{\alpha(h_1, h_2)}\psi_\kappa^{h_1}(\xi_1, w)\psi_\kappa^{h_2}(\xi_2, z)$
 can be seen as the operator-valued distribution multiplied by
 a smooth function $(1-\textstyle{\frac zw})^{\alpha(h_1, h_2)}$ on $I_-\times I_+$.
 
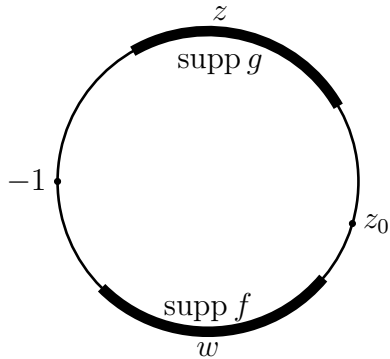
\begin{figure}[ht]
\centering
\begin{tikzpicture}[scale = 0.8]
\clip(-4,-3) rectangle (4,3);
\draw [line width=1pt] (0,0) circle (2.5);

\draw (-2.5,0) node[anchor=east] {$-1$};
\draw (2.4,-0.7) node[anchor=west] {$z_0$};
\draw (0,-2.5) node[anchor=south] {$\supp f$};
\draw (0,-2.5) node[anchor=north] {$w$};
\draw (0.2,2.3) node[anchor=north] {$\supp g$};
\draw (0.2,2.5) node[anchor=south] {$z$};
\draw [fill](-2.5,0) circle (0.05);
\draw [fill](2.4,-0.7) circle (0.05);

\draw [line width=4pt] (2.165063509, 1.25) arc (30:120:2.5) ;
\draw [line width=4pt] (-1.767766953, -1.767766953) arc (-135:-40:2.5) ;

\end{tikzpicture}\caption{The circle, the point of infinity ($-1$) and two functions with disjoint supports.}
\label{fig:braiding}
\end{figure}

 Similarly, $(1-\textstyle{\frac wz})^{\alpha(h_1, h_2)}\psi_\kappa^{h_2}(\xi_1, z)\psi_\kappa^{h_1}(\xi_2, w)$
 can be seen as an operator-valued distribution.
 Therefore, by the braiding relation of the primary fields, we have the equality
 \begin{align*}
 (1-\textstyle{\frac zw})^{\alpha(h_1, h_2)}w^{\alpha(h_1, h_2)}\psi_\kappa^{h_1}(\xi_1, w)\psi_\kappa^{h_2}(\xi_2, z)
 = (1-\textstyle{\frac wz})^{\alpha(h_1, h_2)}z^{\alpha(h_1, h_2)}\psi_\kappa^{h_2}(\xi_1, z)\psi_\kappa^{h_1}(\xi_2, w)  
 \end{align*}
 in the sense of operator-valued distributions, under the restriction on the supports of $f$ and $g$.
 Seen as functions, the quotient of $(1-\textstyle{\frac wz})^{\alpha(h_1, h_2)}z^{\alpha(h_1, h_2)}$
 by $(1-\textstyle{\frac zw})^{\alpha(h_1, h_2)}w^{\alpha(h_1, h_2)}$
 is exactly $\lim_{\Im\zeta > 0, \zeta \to -1}\zeta^{\alpha(h_1, h_2)}$.
 Therefore, by dividing the equation by the former factor, we obtain
 \begin{align*}
 \psi_\kappa^{h_1}(\xi_1, w)\psi_\kappa^{h_2}(\xi_2, z)
 = \epsilon^+_{\kappa(h_1), \kappa(h_2)}\psi_\kappa^{h_2}(\xi_1, z)\psi_\kappa^{h_1}(\xi_2, w)  
 \end{align*}
 as desired, as operator-valued distributions.
\end{proof}

The following is essentially due to \cite[Section 6.3]{CKL08},
where the case $D=\frac12$ is treated, see also \cite[Proposition 6.4]{CKLW18}, \cite{Toledano-Laredo99-1}.
\begin{lemma}\label{lm:chargeddiff}
 Assume that a charged primary field $\psi_\kappa^h$ has conformal dimension $D$ and satisfies polynomial energy bounds as above
 (we omit the dependence on $\xi$).
 Then it is diffeomorphism covariant, that is,
 with $\hat U_\kappa$ the projective unitary representation of $\overline{\diff(S^1)}$, 
 $\Ad \hat U_\kappa(\gamma)(\psi_\kappa^h(f)) = \psi_\kappa^h(\beta_D(\gamma)(f))$.
\end{lemma}
\begin{proof}

 Let us sketch the arguments of \cite[Section 6.3]{CKL08}.
 Note that $\hat L_{\kappa,n}$ is a representation of the Virasoro algebra on the dense domain $\hat \cH_\kappa^\fin$ in the Hilbert space $\hat \cH_\kappa$.
 By polynomial energy bounds they extend to $C^\infty(\hat L_{\kappa,0})$.
 Furthermore, on this domain we have the commutation relations
 \begin{align*}
  i[\hat L_\kappa(f_1), \psi_\kappa^h(f_2)] = \psi_\kappa^h(Df_1'f_2 - f_1f_2').
 \end{align*}
 From this, the following relation follows:
 \begin{align*}
  \frac{d}{dt}\psi_\kappa^h(\beta_D(\exp(tf_1))f_2)|_{t=0} = i[\hat L_\kappa(f_1), \psi_\kappa^h(f_2)].
 \end{align*}
 Again by energy bounds for $\hat L_{\kappa, n}$, it holds that
 $U_\kappa(\exp(tf))$ preserves the domain $C^\infty(L_{\kappa,0})$.
 
 For $\Psi \in C^\infty(\hat L_{\kappa,0})$, we can consider two vector-valued functions of $t\in \bbR$
 \begin{align*}
  \Psi_1(t) &= \psi_\kappa^h(\beta_D(\exp(tf_1)f_2)\hat U_\kappa(\exp(tf_1))\Psi, \\
  \Psi_2(t) &= \hat U_\kappa(\exp(tf_1))\psi_\kappa^h(f_2)\Psi, 
 \end{align*}
 and they both satisfy the differential equation $\Psi'(t) = i\hat L_\kappa(f)\Psi(t)$.
 As they satisfy the same initial condition $\Psi_1(0) = \Psi_2(0) = \psi_\kappa^h(f_2)\Psi$,
 they must coincide. Therefore, we have $\Ad \hat U_\kappa(\exp(tf_1))(\psi(f_2)) = \psi_\kappa^h(\beta_D(\exp(tf_1))f_2)$.
 
 Recall that $\diff(S^1)$ is algebraically simple \cite{Thurston74}
 while the subgroup generated by one-parameter groups is a normal subgroup of $\diff(S^1)$.
 Therefore, they must coincide.
 This implies that any element in $\overline{\diff(S^1)}$ is a product of elements
 in some one-parameter group and an element in the center $\{R_{2\pi n}\}$.

 Now the relation holds for all $\gamma \in \overline{\diff(S^1)}$
 which is a product of elements in one-parameter groups, and it is straightforward
 to verify it for $R_{2\pi}$, therefore, we have the desired covariance.
\end{proof}

\subsection{One-dimensional Wightman fields arising from trivial total braiding}\label{1dwightman}

Let us first assume that there are two conformal nets $\cA_1, \cA_2$ on $S^1$
and charged primary fields $\psi_1, \psi_2$ of these nets
(here we omit the dependence on $h \in G$ and $\xi \in \Xi$).

\begin{lemma}\label{lm:normalproduct-bounds}
 Let $\psi_1(z)\otimes \bb1, \bb1\otimes\psi_2(z)$ be charged primary fields satisfying polynomial energy bounds
 with respect to an operator $\hat L_0\otimes \bb1 + \bb1\otimes \hat L_0$.
 Then $\psi_1(z)\otimes\psi_2(z)$ satisfies polynomial energy bounds.
\end{lemma}
\begin{proof}
 Let us assume polynomial energy bounds:
 there are $r_j, p_j, C > 0$, for $j=1,2$ such that
 \begin{align*}
 \|\psi_j(s)\Psi\| \le C(1+|s|)^{r_j} \|(\hat L_0 + \bb1)^{p_j}\Psi\|
 \end{align*}
 for $\Psi \in C^\infty(\hat L_0)$.
 Note that
 \begin{align*}
  (\hat L_0+\bb1)\psi_2(t-s) &= [\hat L_0, \psi_2(t-s)] + \psi_2(t-s)\hat L_0  + \psi_2(t-s) \\
  &= (s-t + 1)\psi_2(t-s) + \psi_2(t-s)\hat L_0 \\
  &= \psi_2(t-s)(\hat L_0 + (s-t + 1)),
 \end{align*}
 hence $(\hat L_0+\bb1)^q\psi_2(t-s) = \psi_2(t-s)(\hat L_0 + (s-t + 1)\bb1)^q$.
 Then, for a fixed $s$ and $\Psi$ such that $H \Psi = \ell\Psi$ where $H = \hat L_0 \otimes \bb1 + \bb1 \otimes \hat L_0$,
 \begin{align*}
  \|\sum_{t}\psi_1(s-t)\otimes\psi_2(t)\Psi\|
  & \le \sum_{t}C(1+|s-t|)^{r_1} \|(H + \bb1)^{p_1}(\bb1\otimes\psi_2(t))\Psi\| \\
  & \le \sum_{t}C(1+|s-t|)^{r_1} \|(H + \bb1)^{\lceil p_1\rceil}(\bb1\otimes\psi_2(t))\Psi\| \\
  & \le \sum_{t}C(1+|s-t|)^{r_1} \|(\bb1\otimes\psi_2(t))(H + (t+1)\bb1)^{\lceil p_1\rceil}\Psi\| \\
  & = \sum_{t}C(1+|s-t|)^{r_1}(\ell + (t+1))^{\lceil p_1\rceil} \|(\bb1\otimes\psi_2(t))\Psi\| \\
  &\le \sum_{t}C^2(1+|s-t|)^{r_1}(1+|t|)^{r_2}(\ell + (t+1))^{\lceil p_1\rceil} \|(H + \bb1)^{p_2}\Psi\| \\
  &\le C^2(1+|s|)^{2r_1+r_2 + \lceil p_1\rceil} \|(H + \bb1)^{r_1 + r_2 + 2\lceil p_1\rceil+p_2}\Psi\|,
 \end{align*}
 where $\lceil p\rceil$ is the smallest integer larger or equal to $p$,
 in the last line we estimated $1+|s-t| \le 1 + |s| + \ell \le (1+|s|)(1+\ell)$
 because $(\psi_1(s-t)\otimes \bb1)\Psi = 0$ when $t < -s-\ell$ or $(\bb1\otimes\psi_2(t))\Psi = 0$ when $t > \ell$.

 For a general $\Psi = \sum_\ell \Psi_\ell$,
 the $\sum_{t}\psi_1(s-t)\otimes\psi_2(t) \Psi_\ell$ are orthogonal for different $\ell$.
 Therefore, a polynomial energy bound follows for $\psi_1(z)\otimes\psi_2(z)$.
\end{proof}
Clearly, this estimate is not optimal.
It can happen, as we will see with concrete examples, that
a product of two charged primary fields whose Fourier components are bounded $\|\psi(s)\| \le C$
has again bounded Fourier components.

\begin{lemma}\label{lm:normalproduct-covariance}
 Let $\psi_1(z), \psi_2(z)$ be charged primary fields for nets $(\cA_1, U_1, \Omega_1), (\cA_2, U_2, \Omega_2)$
 with polynomial energy bounds.
 Then $\psi_1(z)\otimes \psi_2(z)$ is diffeomorphism covariant
 with respect to $\widetilde U_1 \otimes \widetilde U_2$.
\end{lemma}
\begin{proof}
 By Lemma \ref{lm:normalproduct-bounds}, $\psi_1(z)\otimes \psi_2(z)$ satisfies polynomial energy bounds
 with respect to $H=\hat L_{1,0}\otimes \bb1 + \bb1 \otimes \hat L_{2,0}$.
 Then the diffeomorphism covariance follows from Lemma \ref{lm:chargeddiff}.
\end{proof}

As in Section \ref{net-trivialtotalbraiding}, let $\{(\cA_\kappa, U_\kappa, \Omega_\kappa)\}_{\kappa \in K}$ be a finite family of conformal nets on $S^1$ generated by fields $\{\phi_{\kappa, j}\}_{\kappa \in K}$,
with a collection $\Delta_\kappa$ of localized automorphisms in some interval,
with a common (finitely generated) abelian group structure $G$ and a family of bijections $\kappa:G \to \Delta_\kappa$ for every index $\kappa \in K$.
Let $\{\psi_\kappa^h(\xi,z)\}_{\kappa\in K, h \in G, \xi \in \Xi_\kappa}$ be a family of charged primary fields
satisfying the conditions of Section \ref{charged-fields}.

We assume that $\prod_{\kappa} \epsilon^+_{\kappa(h_1), \kappa(h_2)} = 1$ for all pairs $h_1, h_2 \in G$, as in Section \ref{net-trivialtotalbraiding}.
Let us pick $\xi_\kappa$ for each $\kappa$.
Then we consider the formal series $\psi_K^h(\xi_K,z) := \bigotimes_{\kappa} \psi_\kappa^h(\xi_\kappa, z)$
acting on the auxiliary space $\bigotimes_{\kappa} \hat\cH_\kappa$, where $\xi_K=(\xi_1,\cdots, \xi_K)$ is an index to label the formal series for the tensor product of $K$ factors.
This is a (normal) product of two commuting formal series, hence this makes sense as a formal series,
see Section \ref{formal}. It follows from Lemma \ref{lm:normalproduct-bounds} that it makes sense as an operator-valued distribution
if each of the fields satisfies polynomial energy bounds.

\begin{theorem}\label{th:1dfield}
 The fields $\psi_K^h$, $h \in G$, together with chiral fields $\phi_{\kappa, j}$, is a conformal Wightman field theory
 acting on the Hilbert space $\cH_K = \bigoplus_{g \in G} \bigotimes_{\kappa} \cH_\kappa^{\kappa(g)}$.
 If, in addition, $\psi_K^h$ satisfies linear energy bounds with respect to
 $\sum_\kappa \bb1\otimes \cdots \otimes\bb1\otimes \underset{\kappa\text{-th}}{\hat L_0}\otimes \bb1\otimes \cdots \otimes \bb1$ for all $h\in G$,
 then $\{\psi_K^h: h \in G\}$ generates a conformal net on $S^1$.
\end{theorem}
\begin{proof}
 The formal series $\psi_K^h$ is defined on the Hilbert space $\bigotimes_{\kappa}\hat \cH_\kappa$.
 It generates vectors in  $\cH_K = \bigoplus_{g \in G} \bigotimes_{\kappa} \cH_\kappa^{\kappa(g)}$
 from the vacuum. Then the $\phi_{\kappa,j}(f)$ generate a dense set of vectors in each of these summands.
 On the other hand, no other vector is generated from $\psi_K^h$ and $\phi_{\kappa,j}$ acting on the vacuum.

 Let us prove locality. For two test functions $f_1, f_2$
 with $\arg \supp f_1 < \arg \supp f_2$, we have by Lemma \ref{lm:braiding} that
 \begin{align*}
  \psi_K^{h_1}(f_1)\psi_K^{h_2}(f_2) 
  &=\prod_{\kappa}\epsilon^+_{\kappa(h_1), \kappa(h_2)}
  \psi_K^{h_2}(f_2)\psi_K^{h_1}(f_1) \\
  &=  \psi_K^{h_2}(f_2)\psi_K^{h_1}(f_1),
 \end{align*}
 giving locality as operator-valued distributions.
 The fields $\psi_K^h$ and $\phi_{\kappa,j}$ are relatively local.
 The representation $U_K$ is a tensor product of positive-energy representations,
 therefore it has itself positive energy.
 
 If we assume linear energy bounds, then the fields strongly commute
 and generate a conformal net on $S^1$ by Theorem \ref{th:locality-doublecomm}.
\end{proof}

\begin{remark}\label{rm:strong}
We expect that the conformal net on $S^1$ in Theorem \ref{th:1dfield} generated by the conformal Wightman field is unitarily equivalent to the net extension $\cA_K$ constructed in Theorem \ref{th:1dext}. This, at least for $|K| = 2$, should follow by a similar argument as in the proof of Theorem \ref{th:2dfield}.
More generally, assuming properties that assures that the VOA intertwiners intertwine
DHR sectors (strong integrability, strong intertwining property and strong braiding \cite{CWX, Gui20-2}),
it should be possible to show that the two nets coincide. We will address this question in a future publication.
\end{remark}

\subsection{Two-dimensional Wightman fields through braiding-cancelling map}\label{2dwightman}

From a pair of charged primary fields with the braiding satisfying certain conditions, we can construct
a two-dimensional conformal Wightman field.
This is in particular the case if we take two copies of the same conformal net $\cA_\rL = \cA_\rR$,
with a pointed braided \emph{fusion} subcategory of automorphisms whose fusion ring is isomorphic to a \emph{finite} abelian group $G$,
and we take the tensor product of the left charged fields $\psi_\rL^h$ and $\psi_\rL^{h^{-1}}$, where $h \in G$, similarly to \eqref{eq:combinegandg}, namely $\widetilde \psi^h := \psi_\rL^h \otimes \psi_\rL^{h^{-1}}$ acting on $\cH_\rL \otimes \cH_\rL$.

If we assume that the braidings cancel in the sense\footnote{Here, instead of assuming additionally that
$\epsilon^+_{\rL(g)^{-1}, \rL(h)} = \overline{\epsilon^-_{\rR(g)^{-1}, \rR(h)}}$ as we did in Section \ref{2dnet},
we assume that the set of charged primary fields is closed under conjugate and each of them
satisfies the braiding relation.} that
$\epsilon^+_{\rL(g), \rL(h)} = \overline{\epsilon^-_{\rR(g), \rR(h)}}$,
we can prove locality of the combined Wightman field. We show that this corresponds to a (finite index) Longo--Rehren extension with respect to a pointed braided fusion category, \cite{LR95}, see Section \ref{LRsection}. The basic idea behind these constructions is also present in the two-dimensional quantum field theory context, \eg, in \cite{Rehren97}, see in particular \cite[Section 2.3]{Rehren97}.

Let $\{(\cA_{\rL/\rR}, U_{\rL/\rR}, \Omega_{\rL/\rR})\}$ be a pair of conformal nets on $S^1$ generated by fields $\{\phi_{{\rL/\rR}, j}\}$,
with a collection $\Delta_{\rL/\rR}$ of (mutually commuting inequivalent) localized automorphisms in some interval (and closed under composition and inverses up to unitary equivalence), as in Section \ref{2dnet},
with a common finite abelian group structure $G$ and bijections ${\rL/\rR}:G \to \Delta_{\rL/\rR}$.
Let $\{\psi^h_{\rL/\rR}(\xi,z)\}_{h \in G, \xi \in \Xi_{\rL/\rR}}$ be a family of primary fields
satisfying the conditions of Section \ref{charged-fields}.
Denote $\epsilon^+ := \epsilon^+_{\rL(g), \rL(h)}$ and $\epsilon^- := \epsilon^-_{\rR(g), \rR(h)}$, for short. Then our braiding cancellation assumption reads $\epsilon^+ = \overline{\epsilon^-}$.
For $w_1,w_2 \in S^1 \setminus \{-1\}, \arg w_1 < \arg w_2$, we have by Lemma \ref{lm:braiding} that
\begin{align*}
 \psi^g_{\rL}(\xi_\rL,w_1)\psi^h_{\rL}(\xi_\rL,w_2) &= \epsilon^+\psi^h_{\rL}(\xi_\rL,w_2)\psi^g_{\rL}(\xi_\rL,w_1), \\
 \psi^g_{\rR}(\xi_\rR,w_2)\psi^h_{\rR}(\xi_\rR,w_1) &= \epsilon^- \psi^h_{\rR}(\xi_\rR,w_1)\psi^g_{\rR}(\xi_\rR,w_2) \\
 &= \overline{\epsilon^+}\psi^h_{\rR}(\xi_\rR,w_1)\psi^g_{\rR}(\xi_\rR,w_2).
\end{align*}
For fixed $\xi_\rL$ and $\xi_\rR$, let us introduce a two-dimensional formal power series by
\begin{align*}
 \widetilde\psi^h(w,z) := \psi^h_{\rL}(\xi_\rL,w)\otimes\psi^{h}_{\rR}(\xi_\rR,z).
\end{align*} 

\begin{theorem}\label{th:2dfield}
 The field $\widetilde\psi^h$ is a two-dimensional conformal Wightman field
 on the Hilbert space $\widetilde \cH = \bigoplus_{h \in G} \cH_\rL^{\rL(h)}\otimes \cH_\rR^{\rR(h)}$.
 If the field satisfies a linear energy bound with respect to $\hat L_0\otimes \bb1 + \bb1 \otimes \hat L_0$,
 then it generates a two-dimensional conformal net.
 If $G$ is finite, the net coincides (up to unitary equivalence) with $\widetilde\cA$ constructed in Theorem \ref{th:2dext} in the case $\cA_\rL = \cA_\rR$.
\end{theorem}
\begin{proof}
The field $\widetilde\psi^h$ satisfies two-dimensional locality:
the point $(w_1, z_1)$ is spacelike to $(w_2, z_2)$ if
$\arg w_1 < \arg w_2$ and $\arg z_1 > \arg z_2$ (or the reversed relations, and in this case,
with $\epsilon^+ := \epsilon^+_{\rL(h), \rL(g)}$),
\begin{align*}
 \widetilde\psi^{h}(w_1,z_1)\widetilde\psi^{g}(w_2,z_2)
 &= \psi^{h}_{\rL}(\xi_\rL,w_1)\otimes \psi^{h}_{\rR}(\xi_\rR,z_1)\cdot \psi^{g}_{\rL}(\xi_\rL,w_2)\otimes \psi^{g}_{\rR}(\xi_\rR,z_2) \\
 &= \epsilon^+ \overline{\epsilon^+}\psi^{g}_{\rL}(\xi_\rL,w_2)\otimes \psi^{g}_{\rR}(\xi_\rR,z_2)
    \cdot \psi^{h}_{\rL}(\xi_\rL,w_1)\otimes \psi^{h}_{\rR}(\xi_\rR,z_1) \\
 &= \psi^{g}_{\rL}(\xi_\rL,w_2)\otimes \psi^{g}_{\rR}(\xi_\rR,z_2) \cdot \psi^{h}_{\rL}(\xi_\rL,w_1)\otimes \psi^{h}_{\rR}(\xi_\rR,z_1) \\
 &=  \widetilde\psi^{h_2}(w_2,z_2)\widetilde\psi^{h_1}(w_1,z_1).
\end{align*}
 It is a two-dimensional operator-valued distribution because it is a tensor product
 of two one-dimensional operator-valued distributions.
 Other axioms, positivity of energy, diffeomorphism covariance, the cyclicity of vacuum in $\widetilde\cH$,
 can be proven as in Theorem \ref{th:1dfield}.
 If we assume linear energy bounds, strong commutativity of smeared fields follows.
 
 Let $G$ be finite.
 As the decomposition of the vacuum representation $\widetilde \cH = \bigoplus_{h \in G} \cH_\rL^{\rL(h)}\otimes \cH_\rR^{\rR(h)}$ with respect to the action of $\cA_\rL \otimes \cA_\rR$ is the same as the decomposition of the vacuum representation of the Longo--Rehren extension $\tilde \cA_{\LR}$ \eqref{eq:HilbLRext}, and the former determines the dual canonical endomorphism of the extension we constructed, by \cite[Proposition 3.4]{LR95}, the latter endomorphism must be unitarily equivalent to the Longo--Rehren endomorphism $\Theta_{\LR}$.
 As we are in the braided pointed fusion case with finite abelian group $G$, the uniqueness of the associated Longo--Rehren extension
 (among \textit{finite index} local extensions) provided by \cite[Section 4, Example 4.5]{KL04-2} proves the last claim.
\end{proof}

\section{Examples: the \texorpdfstring{$\rU(1)$}{U(1)}-current}\label{U(1)current}
We construct examples of a conformal net on $S^1$, collections of sectors and charged primary fields
associated with the $\rU(1)$-current. To the best of our knowledge, most of the material of this section is scattered in several places, see, \eg, \cite{BMT88, Toledano-LaredoThesis, TZ12}, or it is known as folklore. We collect some useful facts for the sake of the reader to appreciate the construction of Wightman fields through charged primary fields in the special case of extensions of the U(1)-current,  cf. \cite{Rehren97} for the two-dimensional case. 

Here we use the symbols $\cH, \cA, U, \Omega, L_n, J_n, Y_{\alpha,n}$ for a chiral component,
differently from the previous sections where we used the index $\kappa$.
\subsection{The field and the net}

In literature, several constructions of the $\rU(1)$-current appeared.
One can use the lowest weight representation with the lowest weight $1$ of
$\rP\rS\rL(2,\bbR)$,
and the associated net standard subspaces,
then by second quantization the local algebras of the $\rU(1)$-current net are generated by the Weyl operators, cf. \cite{BGL93}.

The construction of the $\mathrm{U}(1)$-current that we illustrate here relies on the so-called current algebra,
namely the complex Lie algebra generated as $\mathbb{C}$-linear space by the family of symbols $\{J_n:J_n:n\in\mathbb{Z}\}$
and a central element $c$ that verify $[J_m,J_n]=m\delta_{n+m}c$, see \cite{BMT88}.
The vacuum representation of the current algebra is given as a Verma module (with the lowest weight $0$, hence $J_0 \Omega = 0$),
which is the linear span of the symbols
\[
J_{-i_1}\ldots J_{-i_m}\Omega,\quad 0<i_1\leq i_2\leq\ldots\leq i_m\in\mathbb{N}, m\in\mathbb{N} 
\]
where, by abuse of notation, $J_n$ are interpreted as operators with the Lie brackets given for an associative algebra,
$J_k\Omega=0$ for $k\geq0$ and one defines the actions of $J_n$ for $n\in\mathbb{Z}$ by using the commutation relations.
There is a unique inner product satisfying $J_n^* = J_{-n}$ and
\begin{equation}\label{inner_product}
	\langle J_{-i_1}\ldots J_{-i_m}\Omega, J_{-l_1}\ldots J_{-l_s}\Omega\rangle:=\langle \Omega, J_{i_m}\ldots J_{i_1}J_{-l_1}\ldots J_{-l_s}\Omega\rangle,
\end{equation}
and we extend it by linearity. Its completion is isomorphic to
the Bosonic Fock space with the one-particle space spanned by
$\{J_{-n}\Omega : n\in\mathbb{N}\}$, see, \eg, \cite[Proposition 2.1]{KR87}.

Using Sugawara's formula to define new generators $L_n$ as
\[
L_n:=\frac12\sum_{k\in\mathbb{Z}}:J_{k+n}J_k:,\quad n\in\mathbb{Z} 
\]
where $:\cdot:$ is the Wick product, one defines a projective unitary representation of the Virasoro algebra with the central charge $c=1$.
In particular, the \textbf{conformal Hamiltonian} $L_0$ acts as
\[
 L_0J_{-i_1}\ldots J_{-i_m}\Omega=(i_1+\ldots+i_m)J_{-i_1}\ldots J_{-i_m}\Omega,
\]
for $0\leq i_1\leq\ldots\leq i_m\in\mathbb{N}$.
Therefore, $L_0$ extends to a positive self-adjoint operator.
Moreover, as $J_n$ can be regarded as an annihilation operator for $n > 0$, we have the following bound:
for every $n\in\mathbb{Z}$ and every $\Psi$ in the Verma module
\begin{align*}
 \|J_n\Psi\| &\leq(n+1)\|(L_0+1)^{\frac12}\Psi\|, \\
 \|L_n\Psi\| &\leq \sqrt{\frac{13}{12}} (1+|n|^\frac32)\|(L_0+1)\Psi\|.
\end{align*}
From this, the representation $\{L_n\}$ of the Virasoro algebra integrates to a projective unitary representation $U$ of $\diff(S^1)$.

Moreover, the current
\[
 J(f):=\sum_{n\in\bbZ}f_nJ_n,
\]
for a test function $f\in C^{\infty}(S^1)$ is an essentially self-adjoint operator on $\cD(L_0)$.
Therefore, $J(f)$ for $f\in C^{\infty}(S^1)$ is well-defined on the dense Verma module and $J(f)\Psi\in\mathcal{D}(L_0)$.
We have
\[
 \left\|[L_0,J(f)]\Psi\right\|=\|J(f')\Psi\|\leq \alpha_{f'}\|(L_0+1)^\frac12\Psi\|,
\]
so by Theorem \ref{th:locality-doublecomm} is essentially self-adjoint and it can defined on $\mathcal{D}(L_0)$.
In this way, if $f,g\in C^{\infty}(S^1)$, then the commutator $[J(f),J(g)]$ is well-defined on $\mathcal{D}(L_0)$ and
\[
 [J(f),J(g)]\Psi=\int_{S^1}f'(t)g(t)dt\cdot \Psi,\quad\Psi\in\mathcal{D}(L_0).
\]
Moreover, by Theorem \ref{th:locality-doublecomm},
if the support of $f$ and $g$ is disjoint, then $J(f)$ and $J(g)$ commute strongly, \ie, $[e^{itJ(f)},e^{itJ(g)}]=0$. 
One can also show that $\Ad U(\varphi)(e^{iJ(f)})=e^{iJ(f\circ\varphi^{-1})}$.

For any $\lambda \in \bbR$, and denoting the lowest weight vector by $\Omega_\lambda$, one can construct a Verma module
such that the representation $J_n^\lambda$ is of the form explained above, except that $J_0^\lambda \Omega_\lambda = \lambda \Omega_\lambda$.
The same commutation relations and similar estimates hold for $\lambda \in \bbR$.

To every non-empty non-dense open interval $I\subset S^1$, one associates the following von Neumann algebra
\[
 \mathcal{A}_{\rU(1)}(I):=\{e^{iJ(f)}:\mathrm{supp}f\subset I\}''.
\]
The triple $(\cA_{\rU(1)}, U, \Omega)$ is a conformal net on the circle $S^1$.

\subsection{Representations of the \texorpdfstring{$\rU(1)$}{U(1)}-current}
Let us fix an open non-empty non-dense interval $I\subset S^1$ and a function $h\in C^{\infty}(S^1)$
with $\supp h \subset I$. Then, we define a map on the Weyl operators $W(f):=e^{iJ(f)}$ for $f\in C^{\infty}(S^1)$ supported in $I$ in the following way
\[
 \sigma_{h,I}(W(f)):=e^{i\int_{S^1}f(t)h(t)dt}W(f).
\]

The map $\sigma_{h,I}$ can be extended to a representation of the local algebra $\mathcal{A}_{\rU(1)}(I)$ on the Fock space.
Indeed, for every interval $I_1$, $\sigma_{h,I_1}$
is implemented by $\Ad W(H)$, where $H' = h$ on $I_1$ and supported on another interval $\widetilde I_1$.
Therefore, $\sigma_{H,I_1}$ extends to a representation of the local algebra $\mathcal{A}_{\rU(1)}(I_1)$.
If $I_2\subset S^1$ is another interval such that $I_1\subset I_2$, then a function $H$ supported in an interval $\widetilde I_2 \supset \overline I_2$
such that $H'=h$ on $I_2$ and $\sigma_{h,I_1}(W(f)) = \Ad W(H)(W(f))$.
Therefore, one can define a representation $\sigma_h$ of the net $\mathcal{A}_{\rU(1)}$ as the family $\{\sigma_{h,I_1}\}_{I_1}$ for $I_1\subset S^1$ as above,
verifying the compatibility condition $\sigma_{h,I_2}|_{\mathcal{A}_{\rU(1)}(I_1)}=\sigma_{h,I_1}$
whenever $I_1\subset I_2$ are two non-empty non-dense open intervals of $S^1$.

The unitary equivalence class of the representations is determined by the value  $\int_{S^1}h(t)dt$.
In particular, when $\int_{S^1}h(t)dt=0$, then $\sigma_h$ is unitary equivalent to the vacuum representation.
Let $\alpha = \int_{S^1} h(t)dt$. Although the representation $\sigma_{h}$ is defined on the vacuum Hilbert space,
we denote it $\cH_\alpha$ in order to distinguish the representation.
Analogously, the lowest weight vector in $\cH_\alpha$ is denoted by $\Omega_\alpha$.
See \cite{BMT88} for these results.

\subsection{Fusion relations and braiding}

Let $h_1,h_2\in C^{\infty}(S^1)$. Then the maps $\sigma_{h_1}$, $\sigma_{h_2}$ are automorphisms of the $\mathrm{U}(1)$-current net $\mathcal{A}_{\mathrm{U}(1)}$.
Let $I_0\in\mathcal{I}$ and consider the local algebra $\mathcal{A}_{\mathrm{U}(1)}(I_0)$. Then, $\sigma_{h_1, I_0}$, $\sigma_{h_2,I_0}$
are representations of the local algebra $\mathcal{A}_{\mathrm{U}(1)}(I_0)$ such that $\sigma_{h_1, I_0}\circ\sigma_{h_2, I_0}=\sigma_{h_1+h_2, I_0}$.
Since the interval $I_0$ is arbitrary, we have
$$\sigma_{h_1}\circ\sigma_{h_2}=\sigma_{h_1+h_2}.$$

We now compute the braiding for the U(1)-current, which we believe is well-known, but for which we couldn't find any reference. 

Let $I \in \cI$ be fixed.
Consider two functions $h_1, h_2 \in C^\infty(S^1)$ with $\mathrm{supp}(h_1),\mathrm{supp}(h_2)\subset I$
and let $\sigma_{h_1}$ and $\sigma_{h_2}$ be two automorphisms of the net $\mathcal{A}_{\mathrm{U}(1)}$.
Next, let $I_0 \in \cI$ such that $I \subset I_0$
and choose $I_1,I_2 \subset I_0$ in $\mathcal{I}$ such that $I_2$ stays in the future of $I_1$, and $I_1,I_2, I$ are pairwise disjoint.

Let $i=1,2$ be fixed, and let $g_i\in C^{\infty}(S^1)$ such that $\supp g_i\subset I_i$
and $\int_{S^1} g_i(t)dt = \int_{S^1} h_i(t)dt$.
Then there exist $H_i\in C^{\infty}(S^1)$ supported in $I_0$ such that $\mathrm{Ad}W(H_i)\sigma_{h_i}=\sigma_{g_i}$.
Moreover, for the mutual position of $I_1$ and $I_2$, $H_1$ and $H_2$ have to be of the form $H_1'=-h_1$, $H_2'=-h_2$ on $I$,
$H_i'=g_i$ on $I_i$ and a constant elsewhere.
As the conditions on $h_i, g_i$ show, $H_i$ are piecewise constant in the union of intervals $(I\cup I_i)^c\cap I_0$
(see Figure \ref{fig:braiding_net}).
In particular, if the intervals are between $I_0^c$ and $I$ or between $I_0^c$ and $I_c$,
then $H_i$ is defined as $0$, since $H_i$ is continuous.
On the rest, and for the same reason, we set $H_i$ to
$\int_Ih_i(t)dt=\int_{I_i}g_i(t)dt:= \sqrt{2\pi}\alpha_i$, where $\alpha_i$ is the so called \textbf{charge}.  

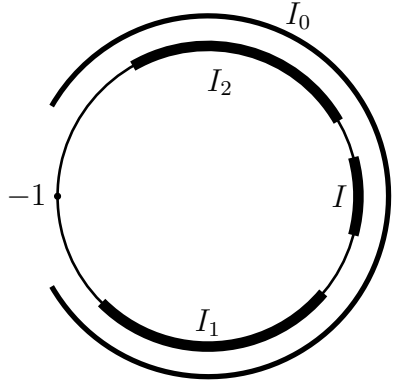
\begin{figure}[ht]
\centering
\begin{tikzpicture}[scale = 0.8]
\clip(-4,-4) rectangle (4,4);
\draw [line width=1pt] (0,0) circle (2.5);

\draw (-2.5,0) node[anchor=east] {$-1$};
\draw (0,-2.5) node[anchor=south] {$I_1$};
\draw (2.5,0) node[anchor=east] {$I$};
\draw (1.5,2.598076211) node[anchor=south] {$I_0$};
\draw (0.2,2.3) node[anchor=north] {$I_2$};
\draw [fill](-2.5,0) circle (0.05);

\draw [line width=4pt] (2.414814566, -0.647047613) arc (-15:15:2.5) ;
\draw [line width=2pt] (-2.598076211, -1.5) arc (-150:150:3) ;

\draw [line width=4pt] (2.165063509, 1.25) arc (30:120:2.5) ;
\draw [line width=4pt] (-1.767766953, -1.767766953) arc (-135:-40:2.5) ;

\end{tikzpicture}\caption{An example configuration for the intervals $I, I_0, I_1$ and $I_0$. The choice of ``future'' and ``past'' is fixed by
considering the point $-1$ as the point of infinity for the lightray,
and choosing $I_0$ that does not contain it.}
\label{fig:braiding_net}
\end{figure}

The braiding is given by the following formula
$$\epsilon(\sigma_{h_1},\sigma_{h_2}) := \epsilon_{\sigma_{h_1},\sigma_{h_2}}=\sigma_{h_2,I_0}(W(H_1)^*)W(H_2)^*W(H_1)\sigma_{h_1,I_0}(W(H_2)).$$
Therefore, we have
\begin{align*}
\epsilon(\sigma_{h_1},\sigma_{h_2})&=\sigma_{h_2,I_0}(W(H_1)^*)W(H_2)^*W(H_1)\sigma_{h_1,I_0}(W(H_2))\\
&=e^{-i\int_{S^1}H_1(t)h_2(t)dt}W(H_1)^*W(H_2)^*W(H_1)e^{i\int_{S^1}h_1(t)H_2(t)dt}W(H_2)\\
&=e^{-i\int_I H_1(t)h_2(t)dt}W(H_1)^*W(H_2)^*W(H_1)e^{i\int_I h_1(t)H_2(t)dt}W(H_2)\\
&=e^{i\int_I \left[h_1(t)H_2(t)-H_1(t)h_2(t)\right]dt}W(H_1)^*W(H_2)^*W(H_1)W(H_2)\\ 
&=e^{i\int_I \left[h_1(t)H_2(t)-H_1(t)h_2(t)\right]dt}e^{-\frac i2\mathrm{Im}\langle H_1,H_2\rangle}W(-H_1-H_2)e^{-\frac i2\mathrm{Im}\langle H_1,H_2\rangle}W(H_1+H_2)\\
&=e^{i\int_I \left[h_1(t)H_2(t)-H_1(t)h_2(t)\right]dt}e^{i\,\mathrm{Im}\langle H_1,H_2\rangle}.
\end{align*}
If one recalls that $\mathrm{Im}\langle H_1,H_2\rangle=\int_{S^1}H_1(t)H'_2(t)dt$, then
\begin{align*}
\int_{S^1}H_1(t)H'_2(t)dt&=\sqrt{2\pi}\alpha_1\int_{I_2}g_2(t)dt+\int_IH_1(t)h_2(t)dt=2\pi\alpha_1\alpha_2+\int_{I\cup I_1}h_1(t)H_2(t)dt.
\end{align*}
Therefore, the braiding is 
\begin{align*}
\epsilon(\sigma_{h_1},\sigma_{h_2})&=e^{i\int_I \left[h_1(t)H_2(t)-H_1(t)h_2(t)\right]dt}e^{i\,\mathrm{Im}\langle H_1,H_2\rangle}\\
&=e^{i\int_I \left[h_1(t)H_2(t)-H_1(t)h_2(t)\right]dt}e^{i2\pi\alpha_1\alpha_2+i\int_{I\cup I_1}h_1(t)H_2(t)dt}\\
&=e^{i\int_I h_1(t)H_2(t)dt}e^{i2\pi\alpha_1\alpha_2}.
\end{align*}
If we take $h_i = \sqrt{2\pi}\alpha_i h$ for a single function $h$ such that $\int_{S^1} h(t)dt = 1$,
then $H_2$ is proportional to a function $H$ on $I$
such that $H' = -h$ on $I$. Therefore, by integration by parts,
$$\int_{I}h_1(t)H_2(t)dt = 2\pi\alpha_1 \alpha_2\int_{I}h(t)H(t)dt=2\pi\alpha_1 \alpha_2\left(\left[-H^2(t)\right]_I-\int_I h(t)H(t)dt\right)$$
implying
$$\int_{I}h_1(t)H_2(t)dt=-\pi\alpha_1\alpha_2.$$
Hence we conclude that in this case
\begin{align*}
\epsilon(\sigma_{\alpha_1 h},\sigma_{\alpha_2 h}) = e^{i\pi\alpha_1 \alpha_2},
\end{align*}
with the configuration above.
In particular, it holds that $\epsilon(\sigma_{\alpha_1 h},\sigma_{\alpha_2 h}) = \epsilon(\sigma_{\alpha_2 h},\sigma_{\alpha_1 h})$.

\subsection{The charged field}

The $\rU(1)$-current net has charged primary fields in the sense of Section \ref{charged-fields}.
We equip $\bbR$ with the scalar product $\<\alpha, \beta\> = \alpha \beta$.
In the sequel, we check that the formal series $Y_\alpha(z) = \sum_{s \in \bbR} Y_{\alpha, s} z^{-s-D}$ are charged primary fields, where
$D = \<\alpha, \alpha\>/2$ and each coefficient $Y_{\alpha, s}$ is a map $\cH_\beta^{\fin} \to \cH_{\beta+\alpha}^{\fin}$
(on each $\cH_\beta^{\fin}, \beta \in \bbR$, only $Y_{\alpha,s}$ with $s \in \bbZ - \<\alpha,\beta\> - D$ are non-zero).
Explicitly, let $c_\alpha$ be the unitary charge shift
operator $\cH_\beta^{\fin} \to \cH_{\beta+\alpha}^{\fin}$ defined
by $c_\alpha J_{-n_1}\cdots J_{-n_k}\Omega_{\beta} = J_{-n_1}\cdots J_{-n_k}\Omega_{\beta + \alpha}$, $n_j > 0$.
Following \cite{TZ12} ($\alpha(n)$ there is identified with $\alpha J_n$ in our notation,
\cf \cite[Chapter V (3.2.1)]{Toledano-LaredoThesis}), we define, as formal series,
\begin{align}
 E^\pm(\alpha, z) 
 &= \exp\left(\mp \sum_{n>0} \frac{\alpha J_{\pm n}}n z^{\mp n}\right), \nonumber \\
 Y_\alpha(z) &= c_\alpha E^-(\alpha,z) E^+(\alpha,z) z^{\alpha J_0}, \label{Eaz} 
\end{align}
where $\alpha J_0 = \<\alpha, \beta\>$ is a scalar on each $\cH_\beta$.
The formal series of the exponential in $E^\pm(\alpha, z)$ is
defined without problems, because, when expanding it into a Taylor series
the coefficients of $z^n$ are finite sums of operators on $\cH_0^{\fin}$.
As for $Y_\alpha(z)$, the coefficients of the product $E^-(\alpha, z) E^+(\alpha, z)$
are infinite sums, but they still make sense as operators. Indeed, for each $n$,
the coefficient $E^+_{\alpha,n}$ of $z^{-n}$ in $E^+(\alpha, z)$ is a linear 
combination of $J_{k_1}\cdots J_{k_j}$
such that $\sum_{m=1}^j k_m = n$, $k_m > 0$, and there are 
only finitely many such combinations.
The series $E^-(\alpha, z)$ has a similar structure.
As any vector $\Psi$ in $\cH_0^{\fin}$ has finite energy $M$ and $J_m$ lowers the energy by $m$,
 $\Psi$ is annihilated by any product $J_{k_1}\cdots J_{k_j}$ if $\sum_m k_m > M$, that is,
 $\Psi$ is annihilated by $E^+_{\alpha,n}$ if $n>M$.
 The coefficient of $z^n$ in $E^-(\alpha,z) E^+(\alpha,z)$ is the 
 sum $\sum_{j\in \bbZ_+} E^-_{\alpha,n-j}E^+_{\alpha,j}$.
 Therefore, on each $\Psi$ only finitely many terms contribute,
 and this sum defines an operator on $\cH_0^{\fin}$.
 Its restriction $Y_\alpha(z)|_{\cH_\beta}$ has the form
 \[
  Y_\alpha(z)|_{\cH_\beta} = \sum_{n \in \bbZ} Y_{\alpha, n - \alpha\beta - D}|_{\cH_\beta} \,z^{-n-D}.
 \]

\paragraph{Braiding}
By \cite[Chapter VI (1.2.2)]{Toledano-LaredoThesis}, it holds that
 \begin{align*}
  	\left(1-\frac{z}w\right)^{-\<\alpha, \beta\>}E^+(\alpha, w)E^-(\beta, z) 
	&= 
	E^-(\beta, z)E^+(\alpha, w) 
 \end{align*}
 where $(1-u)^a = \sum_{n \ge 0}\binom{a}{n}(-u)^n$.
Note that the left-hand side makes sense, because
$(1-\frac{z}w)^{-\<\alpha, \beta\>}$ has positive powers in $z$ and negative powers in $w$,
while $E^-(\beta, z)$ has positive powers in $z$ and $E^+(\alpha, w)$ has negative powers in $w$.
Similarly,
\begin{align*}
 E^-(\alpha, w)E^+(\beta, z) 
 &= \left(1-\frac w{z}\right)^{-\<\alpha, \beta\>}
 E^+(\beta, z)E^-(\alpha, w).
\end{align*}
Therefore, if we introduce the pre-vertex operators
$\underline{Y_{\alpha}}(z) = E^-(\alpha, z) E^+(\alpha, z) z^{\alpha J_0}$
(it is $Y_\alpha(z)$ without $c_\alpha$), and using that $E^\pm(\alpha, z)$ and $E^\pm(\beta, w)$
commute (when $\pm$ coincide), we obtain
\begin{align}
 \left(1-\frac{z}w\right)^{-\<\alpha, \beta\>}\underline{Y_{\alpha}}(w)\underline{Y_{\beta}}(z)
 &=  \left(1-\frac{w}z\right)^{-\<\alpha, \beta\>}
 \underline{Y_{\beta}}(z) \underline{Y_{\alpha}}(w) \, . \label{eq:pre_yy}
\end{align}
Here again, the equality is understood as the equality between the coefficients of
$z^n \zeta^m$ on each fixed vector $\Psi \in \cH_0^{\fin}$.

Now we take $\bigoplus_{\lambda \in \bbR}\cH_\lambda^\fin$ (algebraic direct sum).
Using $c_{\lambda_1} z^{\lambda_2 J_0} = z^{-\<\lambda_1, \lambda_2\> + \lambda_2 J_0}c_{\lambda_1}$,
and the fact that $c_{\lambda_1}$ commute with $E^+(\alpha, z), E^-(\alpha, w)$,
we obtain
\begin{align}
 \left(1-\frac{z}w\right)^{-\<\alpha, \beta\>}w^{-\<\alpha, \beta\>}Y_{\alpha}(w)Y_{\beta}(z)
 &=  \left(1-\frac{w}z\right)^{-\<\alpha, \beta\>}z^{-\<\alpha, \beta\>}
 Y_{\beta}(z) Y_{\alpha}(w) \, . \label{eq:braiding-u1}
\end{align}

\paragraph{Relative locality}
Here again we consider $Y_\alpha(z)$ as a formal series on $\bigoplus_{\lambda \in \bbR} \cH_\lambda^\fin$.
On each of these summands, $L_m$ and $J_m$ act naturally as a representation.
We denote their direct sum on $\bigoplus_{\lambda \in \bbR} \cH_\lambda^\fin$ by $\hat L_m, \hat J_m$,
respectively.

The commutation relation $[\hat J_m, Y_\alpha(z)] = \alpha Y_\alpha(z)z^m$
can be checked easily from the definition of $Y_\alpha(z)$.
This is equivalent to $[\hat J_m, Y_{\alpha, s}] = \alpha Y_{\alpha, m+s}$.

\paragraph{Primarity}
As we will see, the conformal dimension of $Y_\alpha$ should be $D=\frac{\alpha^2}2$. Note that
\begin{align*}
 \partial Y_\alpha(z) 
 &=  z^{-1} Y_\alpha(z)\alpha \hat J_0 
 + \sum_{j<0} z^{-j-1}\alpha Y_\alpha(z)\hat J_j 
 + \sum_{j>0} z^{-j-1}\alpha \hat J_j Y_\alpha(z). 
\end{align*}
On the other hand, $Y_\alpha$ being primary is equivalent to
$[\hat L_m, Y_\alpha(z)] = \partial Y_\alpha(z) z^{m+1} + D(m+1)Y_\alpha(z)$.
Therefore, we need to show
\begin{align}\label{eq:primarygoal}
 &[\hat L_m, Y_\alpha(z)]  \nonumber \\
 &= (z^{-1}  Y_\alpha(z)\alpha \hat J_0 
 + \sum_{j<0} z^{-j-1}\alpha Y_\alpha(z)\hat J_j 
 + \sum_{j>0} z^{-j-1}\alpha \hat J_j Y_\alpha(z)) z^{m+1} 
 + D(m+1) Y_\alpha(z)z^m \, . 
\end{align}
It holds that $\hat J_j c_\alpha = c_\alpha (\hat J_j + \alpha\delta_j)$.
From the Sugawara formula $\hat L_m = \frac12 \sum_k :\hat J_k \hat J_{m-k}:$,
we have $\hat L_m c_\alpha = c_\alpha (\hat L_m + \alpha \hat J_m)$ for $m\neq 0$
and $\hat L_0 c_\alpha = c_\alpha (\hat L_0 + \alpha \hat J_0 + \frac{\alpha^2}2)$. 
 
We want to verify directly that $Y_\alpha(z)$ is primary in the sense above (this should also follow by using VOA modules, see \cite[(36.1)]{CGP21} and references therein, in particular, \cite[Section 5.4]{FHL93}).
To do this, recall the commutation relations
\[
 [\hat L_m, \hat J_j] = -j\hat J_{j+m}.
\]
In the definition of $E^\pm$, the order of product does not matter, hence
\begin{align*}
 E^+(\alpha, z) 
 &= \prod_{j>0} \exp \left(-\tfrac{\alpha \hat J_j}j z^{-j}\right)  \, , 
 \qquad
 E^-(\alpha, z) = \prod_{j<0} \exp \left(-\tfrac{\alpha \hat J_j}j z^{-j}\right) \, . 
\end{align*}
Using the formula $[A, BC] = [A,B]C + B[A,C]$, we calculate
\begin{align}\label{eq:tocompute}
 [\hat L_m, Y_\alpha(z)] &= [\hat L_m, c_\alpha] E^-(\alpha, z)E^+(\alpha, z)z^{\alpha \hat J_0} \\
 &\qquad + c_\alpha [\hat L_m, E^-(\alpha, z)] E^+(\alpha, z)z^{\alpha \hat J_0} +
 c_\alpha E^-(\alpha, z) [\hat L_m, E^+(\alpha, z)]z^{\alpha \hat J_0}\nonumber.
\end{align}

Let us consider the three cases separately. We illustrate our strategy for the simpler case of $m=0$.
\begin{itemize}
\item $m=0$. In this case, \eqref{eq:primarygoal} becomes the following and thus we have to show
\begin{align}
 \label{eq:primaryzero}
 [\hat L_0, Y_\alpha(z)] 
 &= (z^{-1} Y_\alpha(z)\alpha \hat J_0 + \sum_{j>0} z^{-j-1}\alpha \hat J_j Y_\alpha(z)
 + \sum_{j<0} z^{-j-1}\alpha  Y_\alpha(z)\hat J_j) z + D Y_\alpha(z) 
 \nonumber \\
 &= Y_\alpha(z)\alpha \hat J_0 + \sum_{j>0} z^{-j}\alpha \hat J_j Y_\alpha(z)
 + \sum_{j<0} z^{-j}\alpha  Y_\alpha(z)\hat J_j+ D Y_\alpha(z). 
\end{align}
As $[\hat L_0, c_\alpha] = c_\alpha(D + \hat J_0)$,
the first and the last terms in \eqref{eq:primaryzero} are obtained.
By noting $[\hat L_0, \hat J_{j}] = -j\hat J_{j}$, it is straightforward that
\begin{align*}
 \left[\hat L_0, \exp\left(-\tfrac{\alpha \hat J_j}{j} z^{-j}\right)\right]
 = \alpha \hat J_{j} z^{-j}\exp\left(-\tfrac{\alpha \hat J_j}{j} z^{-j}\right) 
\end{align*}
and hence $[\hat L_0, E^\pm(\alpha, z)] = \sum_{\pm j>0} z^{-j}\alpha \hat J_{j}E^\pm(\alpha, z).$
Inserting them in \eqref{eq:tocompute}, they yield the second and the third terms of \eqref{eq:primaryzero}.
 \item $m>0$, odd. In this case, there is no $j$ such that $-2j = m$.	
 Using $[\hat L_m, \hat J_{j}] = -j\hat J_{m+j}$, we have
\begin{align*}
  \left[\hat L_m, \exp\left(-\tfrac{\alpha \hat J_j}{j} z^{-j}\right)\right]
   &= \sum_{k=0}^\infty\tfrac1{k!}\left[\hat L_m, 
   \left(-\tfrac{\alpha \hat J_j}{j} z^{-j}\right)^k\right] \\
   &= \sum_{k=1}^\infty\tfrac1{k!}\tfrac{kj\alpha \hat J_{m+j}}{j} z^{-j}
   \left(-\tfrac{\alpha \hat J_j}{j} z^{-j}\right)^{k-1} \\
   &= \alpha \hat J_{m+j} z^{-j}\sum_{k=1}^\infty \tfrac1{(k-1)!}
   \left(-\tfrac{\alpha \hat J_j}{j} z^{-j}\right)^{k-1} \\  
   &= \alpha \hat J_{m+j} z^{-j}\exp\left(-\tfrac{\alpha \hat J_j}{j} z^{-j}\right) 
 \end{align*}
 and hence $[\hat L_m, E^\pm(\alpha, z)] 
 = \sum_{\pm j>0} z^{-j}\alpha \hat J_{m+j} E^\pm(\alpha, z)$.
  The contribution from the case $j=-m$ in \eqref{eq:tocompute} gives exactly the 
  term $z^{-1}Y_\alpha(z)\alpha \hat J_0z^{m+1}$ in \eqref{eq:primarygoal}.

 If $m+j < -j$, we need to bring the factor $\hat J_{m+j}$ past 
 $\exp\bigl(-\frac{\alpha \hat J_{-m-j}}{-m-j} z^{m+j}\bigr)$.
 Recalling that $[\hat J_j, \hat J_{-j}] = j$, and hence $[\exp (c\hat J_j), \hat J_{-j}] 
 = \sum_{k\ge 0} \frac1{k!}c^k [\hat J_j^k, \hat J_{-j}] = cj\exp(c\hat J_j)$,
 we get a contribution $\alpha^2 z^m Y_\alpha(z)$ for each such $j$.
 There are $\frac{m-1}2$ such cases and it is $(m-1) \frac{\alpha^2}2 z^mY_\alpha(z)$.

 Finally, $[\hat L_m, c_\alpha]E^-(\alpha, z)E^+(\alpha, z)z^{\alpha \hat J_0} = \alpha \hat J_m Y_\alpha(z) = z^{-m-1}\alpha \hat J_m Y_\alpha(z)z^{m+1}$.
 As $m > 0$, we need to bring the factor $\hat J_m$ through $E^-(\alpha, z)$,
 from which there is an additional contribution $\alpha^2 z^mY_\alpha(z)$.
 Altogether, we obtain a contribution of $(m+1) \frac{\alpha^2}2 z^mY_\alpha(z)$
 which is the last term of \eqref{eq:primarygoal}
 and the term corresponding to $j=m$.
 This completes the proof of the case $m>0$ odd.
 
 \item $m>0$, even. In this case, there is $j$ such that $-2j = m$.
 We calculate the commutator $\bigl[\hat L_{-2j}, \exp\bigl(-\frac{\alpha \hat J_j}{j} z^{-j}\bigr)\bigr]$
 slightly modifying the argument above for $m>0$ odd
\begin{align*}
  \left[\hat L_{-2j}, \exp\left(-\tfrac{\alpha \hat J_j}{j} z^{-j}\right)\right]
   &= \sum_{k=0}^\infty\tfrac1{k!}\left[\hat L_{-2j}, 
   \left(-\tfrac{\alpha \hat J_j}{j} z^{-j}\right)^k\right] \\
   &=-\tfrac{\alpha z^{-j}}j[\hat L_{-2j}, J_j]+\sum_{k=2}^\infty\tfrac1{k!}\left[\hat L_{-2j}, \left(-\tfrac{\alpha \hat J_j}{j} z^{-j}\right)^k\right]\\
   &=\alpha z^{-j}J_{-j}+\sum_{k=2}^\infty\tfrac1{k!}\tfrac{(-1)^k\alpha^kz^{-jk}}{j^k}\tfrac{j^2k(k-1)}2J_j^{k-2}\\
   &\quad-\sum_{k=2}^\infty\tfrac1{k!}\tfrac{(-1)^k\alpha^kz^{-jk}}{j^k}kjJ_j^{k-1}J_{-j}\\
   &=\alpha z^{-j}J_{-j}+\tfrac{\alpha^2}2z^{-2j}\sum_{k=2}^\infty\tfrac1{(k-2)!}\tfrac{(-1)^k\alpha^{k-2}z^{-j(k-2)}}{j^{k-2}}J_j^{k-2}\\
   &\quad -\sum_{k=2}^\infty\tfrac1{k!}\tfrac{(-1)^k\alpha^kz^{-jk}}{j^k}kjJ_j^{k-1}J_{-j}\\
   &=\tfrac{\alpha^2}2z^{-2j}\exp\left(-\tfrac{\alpha \hat J_j}{j} z^{-j}\right)+\sum_{k=1}^\infty\tfrac1{(k-1)!}\tfrac{(-1)^{k-1}\alpha^{k-1}z^{-j(k-1)}}{j^{k-1}}J_j^{k-1}J_{-j}\\
   &=\tfrac{\alpha^2}2z^{-2j}\exp\left(-\tfrac{\alpha \hat J_j}{j} z^{-j}\right)+\exp\left(-\tfrac{\alpha \hat J_j}{j} z^{-j}\right)J_{-j}
\end{align*}
 and then moving $\hat J_{j+m} = \hat J_{-j}$, we obtain an additional term of
\[
 \tfrac{\alpha^2}2 z^{-2j} \exp\left(-\tfrac{\alpha \hat J_j}{j} z^{-j}\right) 
 = \tfrac{\alpha^2}2 z^m \exp\left(-\tfrac{\alpha \hat J_j}{j} z^{-j}\right)
\]
and this gives the contribution $\frac{\alpha^2}2 z^m Y_\alpha(z)$.

From $j$ with $m+j < -j$, we get a contribution in \eqref{eq:primarygoal} $\alpha^2 z^m Y_\alpha(z)$ 
for each such $j$ as before, and as $m$ is even there are $\frac{m-2}2$ 
such cases, and together with the contribution 
from the previous paragraph we 
obtain $\frac{m-1}2\alpha^2 z^mY_\alpha(z)$ as before.

 The rest is the same as in the case where $m$ is odd.
 
\item The case $m<0$ is obtained by taking the conjugate
and substituting $\alpha$ by $-\alpha$.
\end{itemize}

\paragraph{Energy bounds}
By \cite[Proposition VI.1.2.1]{Toledano-LaredoThesis},
$\|Y_{\alpha,n}\| \le 1$ if $\<\alpha, \alpha\> \le 1$.
In particular, $Y_\alpha$ satisfies the linear energy bound for such $\alpha$.
It is also shown that for any $\alpha$, $Y_\alpha(z)$ satisfies a polynomial energy bound.

\paragraph{Wightman fields}
Let $G$ be a subgroup of $\bbR$.
Then, for a fixed function $h$ with $\supp h \subset I$,
we can choose a family of automorphisms $\sigma_\alpha$ parametrized by $\alpha \in G$
such that, by omitting the dependence on $h$, $\sigma_\alpha \circ \sigma_\beta = \sigma_{\alpha + \beta}$.
Their braidings satisfy $\epsilon^\pm(\sigma_\alpha, \sigma_\beta) = \epsilon^\pm(\sigma_\beta, \sigma_\alpha) = \epsilon^\pm(\sigma_{-\alpha}, \sigma_{-\beta})$.
Therefore, if we take $\kappa_\rL(\alpha) = \sigma_{\alpha}, \kappa_\rR(\alpha) = \sigma_{-\alpha}$,
the objects satisfy the conditions of Section \ref{2dnet}.
We can also take $\kappa_\rL(\alpha) = \sigma_{\alpha}, \kappa_\rR(\alpha) = \sigma_{\pm \sqrt{\alpha^2 + 2\ell}}$
for any $\ell \in \bbZ$ such that $\alpha^2 + 2\ell > 0$.

Instead, let $G = \bbZ$ and $K = \{1,2,\cdots |K|\}$ be a finite set.
For $j \in K$, we take $\alpha_j$ such that $\sum_{j \in K} \alpha_j^2/2 = 1$,
then $\kappa_j(n) = \sigma_{n\alpha_j}$ satisfy the conditions of Section \ref{net-trivialtotalbraiding},
obtaining extensions of $(\cA_{\rU(1)})^{\otimes j}$ as conformal nets on $S^1$.

On the other hand, we have checked that the charged primary fields $Y_\alpha$ satisfy the conditions of Section \ref{charged-fields}
(without the index $\xi$).
Moreover, if there is $\alpha \in G, |\alpha| \le 1$, then $Y_{\alpha,s}$ are bounded, and hence the two-dimensional Wightman field
as constructed in Section \ref{2dwightman} satisfies a linear energy bound, and generate the conformal net on $\cE$.
If $G$ is a subgroup without such $\alpha$,
the extension given by $G$ can be embedded with a larger net,
where the field satisfy a linear energy bound.
From this, it follows that the fields for $\alpha > 1$ strongly commute when smeared with spacelike separated test functions
by an analogue of \cite[Lemma 3.6]{CTW22}.

\section{Outlook}
This construction should apply also to loop group nets at level $1$ \cite{Wassermann98}, \cite{Toledano-LaredoThesis}.
The bosonic construction of charged primary fields of \cite{Toledano-LaredoThesis} should
give two-dimensional Wightman fields by our construction.

There are a few works on two-dimensional extension of CFT in a language similar to that of vertex operator algebras,
\eg, \cite{HK07FullFieldAlgebras}, \cite{CKM22Gluing}, \cite{Moriwaki23}. We plan to investigate the construction problems of Wightman fields with more generality,
in particular, for loop group nets with higher levels and for Virasoro nets \cite{KL04-2}.
Deforming CFT by a pair of currents \cite{Moriwaki23} should also have a similar realization in the Wightman setting,
and it would be interesting to see whether it has a dynamical meaning, at least in the sense of perturbation theory, cf.\! \cite{CRV22Lorenzian}.

In principle, our fields should be Wick-rotated to an Euclidean theory, that should correspond to
the works above. It would be interesting to understand these Euclidean models in terms of Hilbert spaces
and operators, \cf \cite{FFK89}.


\subsubsection*{Acknowledgements}
We thank
Sebastiano Carpi for interesting discussions on primary fields and for letting us include the proof of Theorem \ref{th:locality-doublecomm},
Yasuyuki Kawahigashi for useful information on 2-cohomology vanishing,
Claudia Pinzari for the reference \cite[(36.1)]{CGP21},
Karl-Henning Rehren for inspiring discussion on braiding,
Slava Rychkov for informing us of the precise statements of \cite{LM75}
and
the Referees for careful reading and useful comments.

M.S.A.\ is a JSPS International Research Fellow and gratefully acknowledges support by
the Grant-in-Aid Kakenhi n.\! 22F21312.
L.G.\ has been supported by the European Union's Horizon 2020 research and innovation programme H2020-MSCA-IF-2017 under Grant Agreement 795151 \emph{Beyond Rationality in Algebraic CFT: mathematical structures and models}.
L.G.\ and Y.T.\ are partially supported by the \emph{MUR Excellence Department Project MatMod@TOV} awarded to the Department of Mathematics, University of Rome Tor Vergata, by the University of Rome Tor Vergata funding \emph{OAQM}, CUP E83C22001800005.
M.S.A., L.G.\ and Y.T.\ are partially supported by GNAMPA--INdAM.

\appendix
\section{On strong locality}

Let us prove a simple criterion for strong commutativity of operators
satisfying linear energy bounds which follows from \cite[Theorem 19.4.4]{GJ87}. The arguments here are due to Sebastiano Carpi.

Let $H$ be a positive self-adjoint operator, $A$ a symmetric operator on $\dom(H)$
and assume that, $\|A\Psi\| \le C\|H\psi\|, \|[H, A]\Psi\| \le C\|H\psi\|$.
We denote
$R(\lambda) = (H + (\lambda + 1)\bb1)^{-1}, R = R(0) = (H + \bb1)^{-1}, \delta(A) = i[H,A],
\delta^k(A) = \underset{k\text{-times}}{\underbrace{\delta(\delta(\cdots(\delta(A)\cdots)))}}$.
By the proof of \cite[Theorem 19.4.1]{GJ87}, we have
\begin{align*}
 R^\frac12 &= \frac 1\pi \int_0^\infty R(\lambda) \lambda^{-\frac12}\,d\lambda, \\
 [A, R^\frac12] &= \frac 1 \pi \int_0^\infty R(\lambda) ((H+(\lambda+1)\bb1)A - A(H+(\lambda+1)\bb1)) R(\lambda) \lambda^{-\frac12}\,d\lambda \\
 &= -\frac i \pi \int_0^\infty R(\lambda) \delta(A) R(\lambda) \lambda^{-\frac12}\,d\lambda.
\end{align*}
Note also that $\|(H + \bb1)R(\lambda)\| = \|R(\lambda)(H + \bb1)\| \le 1$,
$\|R(\lambda)\| \le \frac1{1+\lambda}$
and $\|R^{-1}R(\lambda)\| = \|\frac{H + \bb1}{H + (\lambda + 1)\bb1}\| \le \frac1{\lambda + 1} \le 1$.
Therefore, we find that 
\begin{align*}
 \|[A, R^\frac12]\| &\le \frac 1 \pi \int_0^\infty \frac{\lambda^{-\frac12}}{1+\lambda}\| \delta(A) R(\lambda)\|\,d\lambda \\
 &\le \frac 1 \pi \int_0^\infty \frac{\lambda^{-\frac12}}{1+\lambda}\| \delta(A) R\|\,d\lambda \\
 &= \|\delta(A)R\|.
\end{align*}
This implies that
\begin{align*}
 \|AR - R^\frac 12 A R^\frac 12\| \le \|[A, R^\frac12]\|\cdot \|R^{\frac12} \| \le \|\delta(A)R\|.
\end{align*}
In particular, if $\|AR\|$ is bounded, it follows that
$\|R^\frac 12 A R^\frac 12\| \le \|AR\| + \|\delta(A)R\|$.
If we apply this to $\delta(A)$ and $\delta^2(A) = -[H,[H,A]]$ instead of $A$, we obtain
$\|R^\frac 12 \delta(A) R^\frac 12\| \le \|\delta(A)R\| + \|\delta^2(A)R\|$
and $\|R^\frac 12 \delta^2(A) R^\frac 12\| \le \|\delta^2(A)R\| + \|\delta^3(A)R\|$,
respectively.

Let us cite \cite[Theorem 19.4.3]{GJ87} with $n=1$
and \cite[Theorem 19.4.4]{GJ87}:
\begin{theorem}\label{th:GJ-locality}
The following hold.
\begin{itemize}
 \item 
 Let $H, A, R$ as above, and suppose that
 $R^\frac12 \delta(A) R^\frac12$ and $AR$ are bounded.
 Then $A$ is essentially self-adjoint on any core of $H$.
 \item
 Let $H, A, B, R$ as above, and suppose that
 $R^\frac12 \delta(A) R^\frac12, R^\frac12 \delta(B) R^\frac12, AR, BR$,
 $R^\frac12 \delta^2(A) R^\frac12$, $R^\frac12 \delta^2(B) R^\frac12, \delta(A)R, \delta(B)R$ are bounded.
 Suppose furthermore that $AB, BA$ are defined on $\dom(H)$ and $AB=BA$.
 Then $A,B$ are essentially self-adjoint on any core of $H$
 and their closures commute strongly.
\end{itemize}
\end{theorem}

Combining these observations, we have the following result.
\begin{theorem}\label{th:locality-doublecomm}
The following hold.
\begin{itemize}
 \item 
 Let $H,A$ as above, and suppose that the operators
 $\|A\Psi\| \le C\|(H+\bb1)\Psi\|, \|[H, A]\Psi\| \le C \|(H + \bb1)\Psi\|,
 \|[H,[H, A]]\Psi\| \le C \|(H + \bb1)\Psi\|$.
 Then $A$ is essentially self-adjoint on any core of $H$.
 \item
 Let $H, A, B, R$ as above, and suppose that the operators
 \begin{itemize}
  \item $A, \delta(A) = i[H, A], \delta^2(A) = -[H, [H, A]], \delta^3(A) = -i[H,[H,[H,A]]]$
  \item $B, \delta(B) = i[H, B], \delta^2(B) = -[H, [H, B]], \delta^3(B) = -i[H,[H,[H,B]]]$
 \end{itemize}
 are defined on $\dom(H)$ and
 \begin{itemize}
 \item $\|A\Psi\| \le C\|(H+\bb1)\Psi\|, \|\delta^k(A)\Psi\| \le C \|(H + \bb1)\Psi\|, k=1,2,3$
 \item $\|B\Psi\| \le C\|(H+\bb1)\Psi\|, \|\delta^k(B)\Psi\| \le C \|(H + \bb1)\Psi\|, k=1,2,3$
 \end{itemize}
 Suppose furthermore that $AB, BA$ are defined on $\dom(H)$ and $AB=BA$.
 Then $A,B$ are essentially self-adjoint on any core of $H$
 and their closures commute strongly.
\end{itemize}
\end{theorem}
\begin{proof}
 By the hypothesis $\|A\Psi\| \le C\|(H+\bb1)\Psi\|, \|\delta(A)\Psi\| \le C \|(H + \bb1)\Psi\|$,
 we have that $\|AR\| \le C, \|\delta(A)R\| \le C$, therefore, by the observation above, $\|R^\frac12 A R^\frac12\| \le 2C$.
 Applying the same argument to $\delta(A)$ with the hypothesis $\|\delta^2(A)\Psi\| \le C \|(H + \bb1)\Psi\|$,
 we obtain $\|R^\frac12 \delta(A)R^\frac12\| \le 2C$.
 This and the first assertion of Theorem \ref{th:GJ-locality} (with $C$ replaced by $2C$) complete the proof of the first assertion.
 
 As for the second assertion, we use $\|\delta^3(A)R\| \le C$
 to infer that the boundedness of $R^\frac12 \delta^2(A)R^\frac12$, which is in the assumption of
 the second assertion of Theorem \ref{th:GJ-locality}.
 We have analogously the bounds for operators involving $B$, therefore,
 the second assertion of Theorem \ref{th:GJ-locality} applies.
\end{proof}

This should be compared with some other formulations of the commutator theorem,
\eg, \cite[Theorem X.37]{RSII} which assumes
that $\|A\psi\| \le C\|H\psi\|$ and $|\<\psi, [H,A]\psi\>| \le C\<\psi, H\psi\>$
(these assumptions are very similar to that of Theorem \ref{th:GJ-locality}, the first assertion)
and proves that $A$ is essentially self-adjoint on any core of $H$.
In general, for a closable operator $B$,
$\|B\psi\| \le C\|H\psi\|$ for all $\psi$ does not imply\footnote{
A counterexample on $\bbC^2$ is
$H= \left(\begin{array}{cc} 1 & 0 \\ 0 & 4\end{array}\right),
B= \frac1{\sqrt2}\left(\begin{array}{cc} 1 & 4  \\ -1 & 4 \end{array}\right) = \frac1{\sqrt 2}\left(\begin{array}{cc} 1 & 1 \\ -1 & 1\end{array}\right)H,
\psi= \left(\begin{array}{c} 1 \\ 1 \end{array}\right)$.} $|\<\psi, B\psi\>| \le C\<\psi, H\psi\>$,
therefore, we cannot infer the essential self-adjointness of $A$
from just from $\|A\Psi\| \le C\|H\Psi\|, \|[H,A]\Psi\| \le C\|(H+\bb1)\Psi\|$, but
we need a bound on the higher commutators.

It is also possibile to have weaker assumptions on the commutators, for example assuming
bounds as quadratic forms \cite{DF77}, but then one must be careful with the domains, see \cite[Appendix C]{Tanimoto16-1}.

{\small
\def\polhk#1{\setbox0=\hbox{#1}{\ooalign{\hidewidth
  \lower1.5ex\hbox{`}\hidewidth\crcr\unhbox0}}} \def\cprime{$'$}
}
\end{document}